\documentclass[11pt]{article}

\usepackage[margin=1in]{geometry}
\usepackage{amsmath,amssymb,amsthm,mathtools,bm}
\usepackage{graphicx}
\usepackage{booktabs}
\usepackage{array}
\usepackage{microtype}
\usepackage{enumitem}
\usepackage{url}
\usepackage[colorlinks=true,linkcolor=blue!45!black,citecolor=blue!45!black,urlcolor=blue!55!black]{hyperref}
\usepackage[nameinlink,capitalize]{cleveref}
\hypersetup{pdftitle={Rare Fluctuations from Normally Hyperbolic Invariant Manifolds},pdfauthor={Stephen Wiggins}}
\usepackage{caption}
\usepackage{subcaption}
\usepackage{xcolor}

\graphicspath{{figures/}}
\allowdisplaybreaks

\newtheorem{theorem}{Theorem}[section]
\newtheorem{proposition}[theorem]{Proposition}

\newtheorem{corollary}[theorem]{Corollary}
\newtheorem{assumption}[theorem]{Assumption}
\theoremstyle{definition}

\theoremstyle{remark}
\newtheorem{remark}[theorem]{Remark}

\newcommand{\R}{\mathbb{R}}
\newcommand{\T}{\mathbb{T}}

\newcommand{\C}{\mathcal{C}}
\newcommand{\Lagr}{\Lambda}
\newcommand{\Z}{\mathcal{Z}}
\newcommand{\Hs}{H_{\mathrm{FW}}}
\newcommand{\phit}{\phi_t}

\newcommand{\eps}{\varepsilon}
\newcommand{\dd}{\mathrm{d}}
\newcommand{\cT}{\mathsf{T}}

\newcommand{\rank}{\operatorname{rank}}
\newcommand{\diag}{\operatorname{diag}}

\newcommand{\doi}[1]{\href{https://doi.org/#1}{doi:#1}}

\title{\textbf{Rare Fluctuations from Normally Hyperbolic Invariant Manifolds}}
\author{
Stephen Wiggins\\[0.5em]
\small Hetao Institute of Mathematics and Interdisciplinary Sciences, Shenzhen, China\\
\small School of Mathematics, University of Bristol, Bristol, United Kingdom
}
\date{August 2026}

\begin{document}
\maketitle

\begin{abstract}
Freidlin--Wentzell theory converts weak-noise large deviations into a Hamiltonian variational problem. We study how a $k$-dimensional normally hyperbolic invariant manifold (NHIM) $N$ of the deterministic dynamics appears in this Hamiltonian system. Its zero-momentum copy $N_0=N\times\{0\}$ is invariant, but the Hamiltonian dynamics has $2k$ center directions near $N_0$: $k$ tangent to $N$ and $k$ conjugate covector directions. We construct the resulting local symplectic geometry and show that fluctuation extremals approaching $N_0$ backward in time at the strong normal rate form an $n$-dimensional exact Lagrangian invariant manifold carrying a single-valued action. A counterexample shows that a corresponding zero-energy section need not be normally hyperbolic within $H_{\mathrm{FW}}^{-1}(0)$. We then consider reaction dynamics. If a parameter moves a deterministic trajectory toward a codimension-one reactivity boundary, the minimum Freidlin--Wentzell action required to reach the boundary is quadratic in the distance from the threshold parameter. Its coefficient is determined by the relative motion of trajectory and boundary and by how effectively the available noise acts transversely. This deterministic boundary is distinct from a noise-dependent stochastic transition state or a committor surface. In a solvent--solute model, varying solvent mass moves the phase-space reactivity boundary while leaving the potential-energy surface fixed, changing the rare-event cost without changing the potential-energy barrier.
\end{abstract}



\section{Introduction}

Freidlin--Wentzell theory \cite{FreidlinWentzell2012} has an especially useful feature from the point of view of dynamical systems: in the weak-noise limit, probabilities of rare departures from deterministic motion can be studied through a variational problem.  Consider a diffusion on a smooth $n$-dimensional state space $M$,
\begin{equation}
\dd X_t^\eps=b(X_t^\eps)\,\dd t+\sqrt{\eps}\,\sigma(X_t^\eps)\,\dd W_t,
\qquad a(x)=\sigma(x)\sigma(x)^{\cT}.
\label{eq:sde-intro}
\end{equation}
Here $b$ generates the deterministic motion, $\sigma$ specifies the directions and amplitudes of the noise, and $\eps$ is the small noise-strength parameter.  The symmetric positive-semidefinite matrix $a(x)=\sigma(x)\sigma(x)^{\cT}$ appearing in \eqref{eq:sde-intro} is the diffusion, or local covariance, matrix.  Over a short time interval $\dd t$, the covariance of the stochastic increment is $\eps a(x)\,\dd t$.  Thus $a$ records the directions, correlations, and relative strengths of the noise.  It will enter directly into both the Freidlin--Wentzell action and its Hamiltonian formulation.

It is useful to say at the outset what we mean by a \emph{rare fluctuation path}.  The stochastic process $X_t^\eps$ has random sample paths.  Freidlin--Wentzell theory associates with a sufficiently regular deterministic path $x(t)$ an action $I[x]$ that measures, on the logarithmic small-noise scale, how unlikely it is for a sample path to remain close to $x(t)$ \cite{FreidlinWentzell2012}.  A deterministic trajectory satisfying $\dot x=b(x)$ has zero action.  A path with $\dot x\ne b(x)$ requires an effective noise forcing.  If that departure can be produced by the available noise directions, its action is positive; if it cannot, its action is infinite.  Heuristically, a small tube about such a path has probability of order $\exp[-I[x]/\eps]$.  The path is therefore \emph{rare} because its probability is exponentially small as $\eps\to0$, and it is a \emph{fluctuation} path because its departure from the deterministic drift represents the effect of the noise.  The path $x(t)$ itself is not random once it has been specified.  We will call stationary or minimizing paths of the action \emph{fluctuation extremals}; in the Hamiltonian formulation these are represented by trajectories $(x(t),p(t))$.

Extremals of the action satisfy Hamilton's equations for
\begin{equation}
\Hs(x,p)=p\cdot b(x)+\frac12 p^{\cT}a(x)p.
\label{eq:fw-hamiltonian-intro}
\end{equation}
The new variable $p$ has $n$ components and is conjugate to $x$.  Thus the original $n$-dimensional stochastic state space gives rise to a $2n$-dimensional Hamiltonian phase space with coordinates $(x,p)$.  The variable $p$ is an auxiliary momentum arising from the variational problem; it need not be a physical mechanical momentum.  The original deterministic dynamics is recovered by setting $p=0$.  Nonzero values of $p$ occur on Hamiltonian extremals representing departures from deterministic motion; the minimizing extremals are the large-deviation paths that dominate the corresponding rare event.

Geometrically, $(x,p)$ belongs to the cotangent bundle $T^*M$.  For each $x\in M$, the possible covectors $p$ form the cotangent space $T_x^*M$, called the fiber over $x$.  A section of $T^*M$ assigns one covector to each base point; the choice $p=0$ at every $x$ is the zero section.  We will use this language only when it helps to describe how the invariant structures fit together.  Standard references for the differential-geometric terminology used below include Abraham, Marsden, and Ratiu \cite{AbrahamMarsdenRatiu1988} and Lee \cite{Lee2012}; for the Hamiltonian and symplectic background see also Marsden and Ratiu \cite{MarsdenRatiu1999}.

For attracting equilibria, the geometry generated by minimum-action fluctuation paths has been studied in considerable detail.  Ray methods for singularly perturbed diffusion equations were developed by Cohen and Lewis \cite{CohenLewis1967} and used by Ludwig \cite{Ludwig1975}, who interpreted the rays in the random-perturbation problem as paths of maximum likelihood.  Dykman and Krivoglaz developed an optimal-fluctuation approach to noise-induced transitions between stable states \cite{DykmanKrivoglaz1979}.  Graham and T\'el used the weak-noise Hamiltonian formulation to investigate nonequilibrium analogues of potential functions \cite{GrahamTel1984,GrahamTel1985}.  These and related approaches became important tools in the study of irreversible escape \cite{MaierStein1993PRE,MaierStein1993PRL}.

For a deterministic attractor $A$, the minimum action required to reach a state $x$ is called the \emph{quasipotential}.  The terminology reflects the fact that this minimum-action function plays a role analogous to a potential even when the drift is not generated by the gradient of an ordinary potential \cite{FreidlinWentzell2012}.  More than one extremal path can sometimes reach the same state.  Each path then carries its own action value, while the quasipotential takes the smallest of the competing values.  A fold in the projection of the extremal family to the original state space produces a caustic; a switching set occurs where two minimizing action branches exchange which one has the smaller value.  Caustics, switching sets, nondifferentiable quasipotentials, and related symmetry-breaking phenomena were studied by Graham and T\'el, Jauslin, Maier and Stein, Dykman and collaborators, and Smelyanskiy, Dykman, and Maier \cite{GrahamTel1985,Jauslin1987,MaierStein1996,MaierStein1997,SmelyanskiyDykmanMaier1997}.

The first question of this paper is what changes when the deterministic set from which the rare fluctuation is considered is not an equilibrium point but a compact normally hyperbolic invariant manifold $N\subset M$.  Let $\dim N=k$.  The obvious copy of $N$ in the Freidlin--Wentzell phase space is
\begin{equation}
N_0=N\times\{0\}\subset T^*M.
\label{eq:N0-intro}
\end{equation}
It is invariant because $p=0$ is invariant.  It is not, however, the whole neutral geometry near $N_0$.  At each point of $N$ there are $k$ tangent directions.  In the Hamiltonian system these are accompanied by $k$ conjugate covector directions.  Consequently the directions that are neither strongly contracting nor strongly expanding have dimension $2k$, not $k$.  One aim of the paper is to determine how these $2k$ directions fit together into an actual invariant manifold near $N_0$.

This leads to a local cotangent-space picture.  We show that the $2k$ neutral directions form a symplectic family of subspaces along $N_0$ and that, when the normal contraction and expansion dominate the tangential motion sufficiently strongly, they are tangent to a $2k$-dimensional locally invariant symplectic manifold.  Here ``local'' means in a sufficiently small neighborhood of $N_0$.  Near $N_0$, this invariant manifold has the same symplectic structure as a neighborhood of the zero section in $T^*N$.  This does not mean that an arbitrarily chosen copy of $T^*N$ inside $T^*M$ is invariant, nor does it produce a unique global center manifold.

The local center manifold is useful for understanding the Hamiltonian geometry, but a rare-event calculation needs the Hamiltonian fluctuation extremals associated with $N_0$.  Because $N_0$ is invariant, a trajectory whose initial condition lies on $N_0$ at a finite time can never leave it.  The relevant trajectories are instead those that approach an orbit in $N_0$ as $t\to-\infty$ at the strong normal rate and move away from its neighborhood in forward time along the strongly unstable directions.  They are \emph{backward asymptotic} to $N_0$, not trajectories that start on $N_0$ and then depart.  Taking the union of their strong-unstable fibers over $N_0$ gives $W^{uu}(N_0)$; the notation $uu$ will be defined precisely when this manifold is constructed.  We prove that this union has dimension $n$, exactly half the dimension of the Freidlin--Wentzell phase space, and that the action accumulated along its trajectories defines a single-valued function on it.  In symplectic terminology, these properties say that $W^{uu}(N_0)$ is an exact Lagrangian manifold.  The terminology is introduced only after the corresponding geometric statements have been established.

A separate issue is whether the stability of the deterministic NHIM is automatically inherited by invariant structures of the Freidlin--Wentzell Hamiltonian system away from $p=0$.  It is not.  A nonzero conjugate momentum changes the motion in the base variables $x$, so a Hamiltonian trajectory can sample the transverse stability rates differently from the deterministic trajectory on $N_0$.  We give an explicit example in which $N$ is uniformly normally attracting but a compact regular zero-energy section of the corresponding Hamiltonian center dynamics is not normally hyperbolic within the ambient zero-energy hypersurface $H_{\mathrm{FW}}^{-1}(0)$.  Thus deterministic normal hyperbolicity organizes the local geometry near $N_0$, but additional rate information is needed to control a larger zero-energy invariant set away from $N_0$.

Normal hyperbolicity supplies the mathematical framework for these constructions.  Persistence and stable and unstable manifold theory for NHIMs was developed by Fenichel and by Hirsch, Pugh, and Shub \cite{Fenichel1971,Fenichel1974,HirschPughShub1977}; geometric treatments include \cite{Wiggins1994,Eldering2013}.  Center manifolds over smooth invariant manifolds and compact invariant sets were constructed by Chow, Liu, and Yi \cite{ChowLiuYi2000Smooth,ChowLiuYi2000Sets}; related integrability questions for partially hyperbolic center structures were studied by Bonatti and Crovisier \cite{BonattiCrovisier2016}, and Hamiltonian center-manifold reduction by Mielke \cite{Mielke1991}.  Global fiber structure of invariant manifolds is discussed in \cite{ElderingKvalheimRevzen2018}.  Weinstein's theorem supplies the local symplectic model of a neighborhood of a Lagrangian submanifold by a neighborhood of the zero section of its cotangent bundle \cite{Weinstein1971}.  Stable and unstable invariant manifolds in Hamiltonian systems are also closely related to exact Lagrangian geometry and generating functions \cite{RudnevWiggins1999}.

Large-deviation problems have also been studied when the relevant deterministic invariant set is more complicated than an equilibrium.  Stable limit cycles exhibit logarithmic oscillations in escape rates, switching structures, and nontrivial families of optimal paths \cite{MaierStein1996PRL,SmelyanskiyDykmanMaier1997}.  Quasipotentials for nongradient systems have been computed for equilibria and cycles \cite{Cameron2012}, and local stochastic sensitivity and quadratic quasipotential approximations have been developed near invariant tori \cite{BashkirtsevaRyashko2016}.  Poquet showed that, for systems obtained by a sufficiently small perturbation of a gradient flow possessing an attracting curve, the escape problem can under appropriate hypotheses be reduced to a one-dimensional variational problem along that curve \cite{Poquet2014}.  The invariant manifolds studied here belong to the auxiliary deterministic Freidlin--Wentzell Hamiltonian system.  They should therefore be distinguished from sample-dependent random invariant manifolds of the original stochastic system \cite{Arnold1998,Boxler1989}.

The second question of the paper arises from a different use of invariant phase-space geometry in reaction dynamics.  The periodic-orbit foundations of dynamical transition-state theory for low-dimensional systems were developed by Pechukas, Pollak, and collaborators \cite{PechukasMcLafferty1973,PollakPechukas1978,PechukasPollak1979}; the later work of Wiggins, Wiesenfeld, Jaff\'e, Uzer, and collaborators generalized the construction to higher-dimensional NHIMs and their stable and unstable manifolds \cite{WigginsWiesenfeldJaffeUzer2001,UzerEtAl2002,WaalkensSchubertWiggins2008}.  Kawai, Komatsuzaki, Nagahata, and collaborators developed the related language of reaction and reactivity boundaries \cite{KawaiKomatsuzaki2010,NagahataEtAl2013a,NagahataEtAl2013b}.  Deterministic time dependence, stochastic transition-state theory, and committor-based transition-path theory lead to other moving or probabilistic structures; because these objects answer different questions, their distinctions and the corresponding literature are reviewed in \Cref{sec:switching-law}.

The present question is a weak-noise large-deviation question connecting deterministic reactivity geometry with the Freidlin--Wentzell variational principle.  Consider a family of deterministic systems depending on a parameter $\mu$, together with a corresponding family of initial conditions $x_0(\mu)$.  Let $\mathcal W_\mu$ denote a smooth codimension-one deterministic reactivity boundary for the system at parameter value $\mu$.  Both the initial condition and the boundary may move as $\mu$ varies.  We suppose that at a critical value $\mu_c$ the unforced trajectory lies on the corresponding boundary and that the relative displacement of the trajectory and boundary changes transversely with $\mu$.  No deterministic trajectory crosses an invariant stable manifold for a fixed value of $\mu$; what changes with $\mu$ is the relative position of two parameter-dependent families.

For $\mu$ close to $\mu_c$, we ask for the minimum Freidlin--Wentzell action needed for a weak fluctuation to reach the deterministic reactivity boundary.  This is deliberately not a construction of a stochastic transition state and not a computation of an isocommittor surface.  It is the least large-deviation cost of reaching a specified deterministic geometric separator.  We prove that this cost is quadratic:
\[
S_T^{\rm loc}(\mu)
=\frac{a_T^2}{2Q_T}(\mu-\mu_c)^2+O(|\mu-\mu_c|^3).
\]
The coefficient separates two effects.  The quantity $a_T$ measures the first-order relative motion of the deterministic trajectory and the moving boundary as the parameter changes.  The quantity $Q_T$ measures how effectively the available forcing can move the trajectory in the boundary-normal direction; in linear-control language it is a scalar projection of the controllability Gramian \cite{Sontag1998}.  The same large-deviation setting underlies minimum-action methods developed by E, Ren, Vanden-Eijnden and by Heymann and Vanden-Eijnden \cite{ERenVandenEijnden2004,HeymannVandenEijnden2008}, but the local parameter-dependent boundary-crossing law derived here is the specific question of interest.

The principal application is the solvent--solute Hamiltonian introduced by Garcia-Meseguer and Carpenter and subsequently studied in phase space by Garcia-Meseguer, Carpenter, and Wiggins \cite{GarciaMeseguerCarpenter2019,GarciaMeseguerCarpenterWiggins2019}.  The potential-energy surface is held fixed while the effective solvent mass is varied.  Changing the mass changes the kinetic part of the dynamics and moves the deterministic transition-state geometry even though the potential-energy saddle does not move.  This makes it possible to compare a fixed potential-energy barrier with a moving phase-space reactivity boundary.  Weak noise is then added only to ask for the minimum action required to reach that deterministic boundary.  The calculation is therefore conceptually different from stochastic transition-state theory, where the transition-state geometry itself depends on the noise realization, and from a committor calculation, where the separator is probabilistic.

The organization follows these questions.  \Cref{sec:fw} introduces the Freidlin--Wentzell action and Hamiltonian and states the normal-hyperbolicity assumptions.  \Cref{sec:carrier} explains how the deterministic stable, tangent, and unstable directions appear in the Hamiltonian variational dynamics and constructs the local symplectic center geometry.  \Cref{sec:lagrangian} constructs the strong-unstable invariant manifold of fluctuation extremals that approach $N_0$ in backward time at the strong normal rate and relates its action to the quasipotential.  \Cref{sec:counterexample} shows that deterministic normal hyperbolicity need not make the corresponding zero-energy section normally hyperbolic within $H_{\mathrm{FW}}^{-1}(0)$.  \Cref{sec:exact-torus} gives one short solvable example in which the fluctuation manifold and its action can be written explicitly.  \Cref{sec:switching-law} places the reaction problem in the context of deterministic reactivity boundaries, stochastic transition-state theory, and committor-based transition-path theory, and then proves the quadratic action law.  \Cref{sec:solvent-application} applies that result to solvent-inertia relocation of a deterministic phase-space reactivity boundary.  The appendices collect reproducibility details and a dimension summary.

\section{Freidlin--Wentzell Hamiltonian dynamics}
\label{sec:fw}

We now state the Freidlin--Wentzell construction more precisely.  We work on $M=\R^n$, or in local coordinates on a smooth manifold.  To allow the number of independent noise directions to differ from the dimension of the state space, we write the noise-amplitude matrix as $B(x)\in\R^{n\times m}$.  Its columns specify the directions in state space in which the noise acts, and
\begin{equation}
a(x)=B(x)B(x)^{\cT}
\label{eq:diffusion-matrix}
\end{equation}
is the $n\times n$ diffusion matrix introduced in \eqref{eq:sde-intro}.  It is the covariance matrix of the noise increments after the overall factor $\eps$ is removed, and it records which combinations of state variables are directly forced.  The matrix $a$ need not be invertible.  When it is singular, the noise acts directly in only part of the state space; this is commonly called a degenerate diffusion or degenerate noise.  The solvent example below has this form because the random force acts only on the solvent momentum.

For a possibly singular diffusion matrix, it is useful to formulate the Freidlin--Wentzell action directly in terms of a control \cite{FreidlinWentzell2012}.  The control $u(t)\in\R^m$ represents the departure from the deterministic motion, and its squared $L^2$ norm measures the large-deviation cost.  For a prescribed absolutely continuous path $x:[0,T]\to M$, define
\begin{equation}
v_x(t)=\dot x(t)-b(x(t)).
\label{eq:path-departure}
\end{equation}
The path can be generated by the controlled system precisely when
\begin{equation}
B(x(t))u(t)=v_x(t)
\label{eq:path-control-constraint}
\end{equation}
for almost every $t$.  Its action is
\begin{equation}
I_{[0,T]}[x]
=
\inf_{u\in L^2([0,T];\R^m)}
\left\{
\frac12\int_0^T |u(t)|^2\,\dd t:
B(x(t))u(t)=\dot x(t)-b(x(t))
\right\}.
\label{eq:control-action}
\end{equation}
If no square-integrable control satisfies the constraint, the action is $+\infty$.  Equation~\eqref{eq:control-action} makes explicit how the prescribed path enters the minimization: both the required departure $\dot x-b(x)$ and the available forcing directions $B(x)$ are evaluated along that path.

When $a(x)$ is positive definite, the minimum-norm control satisfying \eqref{eq:path-control-constraint} is
\begin{equation}
u_*(t)=B(x(t))^{\cT}a(x(t))^{-1}
\bigl[\dot x(t)-b(x(t))\bigr].
\label{eq:min-control-full-rank}
\end{equation}
Substituting this expression into \eqref{eq:control-action} gives the familiar formula
\begin{equation}
I_{[0,T]}[x]
=\frac12\int_0^T
(\dot x-b(x))^{\cT}a(x)^{-1}(\dot x-b(x))\,\dd t.
\label{eq:action}
\end{equation}
Equation~\eqref{eq:action} is therefore a convenient special form available only when $a$ is invertible.  When $a$ is singular we do not use $a^{-1}$; the control formulation \eqref{eq:control-action} remains the definition of the action.

The Hamiltonian can also be obtained without assuming that $a$ is invertible.  Introduce an independent conjugate variable $p\in T_x^*M$ and maximize over the control:
\begin{equation}
\Hs(x,p)
=
\sup_{u\in\R^m}
\left[
 p\cdot\bigl(b(x)+B(x)u\bigr)-\frac12|u|^2
\right].
\label{eq:hamiltonian-control-sup}
\end{equation}
The maximizing control is
\begin{equation}
u_*=B(x)^{\cT}p,
\label{eq:optimal-control-p}
\end{equation}
and hence
\begin{equation}
\Hs(x,p)=p\cdot b(x)+\frac12|B(x)^{\cT}p|^2
=p\cdot b(x)+\frac12p^{\cT}a(x)p.
\label{eq:legendre}
\end{equation}
No inverse diffusion matrix has been introduced.  If $a$ is positive definite, the Hamiltonian relation $\dot x=b(x)+a(x)p$ can be inverted to give $p=a^{-1}(\dot x-b)$.  If $a$ is singular, this inversion is not generally possible or unique; $p$ is instead treated as the independent conjugate variable appearing in \eqref{eq:hamiltonian-control-sup}.

Hamilton's equations are
\begin{align}
\dot x&=b(x)+a(x)p,
\label{eq:ham-x}\\
\dot p&=-Db(x)^{\cT}p-\frac12\nabla_x\bigl(p^{\cT}a(x)p\bigr).
\label{eq:ham-p}
\end{align}
Two standard geometric objects associated with a Hamiltonian system will be needed later \cite{MarsdenRatiu1999}.  The canonical one-form and symplectic form on $T^*M$ are
\begin{equation}
\vartheta=p\cdot\dd x,
\qquad
\omega=-\dd\vartheta=\sum_{j=1}^n \dd x_j\wedge\dd p_j.
\label{eq:canonical-forms}
\end{equation}
The symplectic form will be used to characterize the center and fluctuation manifolds.  The canonical one-form will later connect the fluctuation manifold to the Freidlin--Wentzell action.

The zero section
\begin{equation}
\Z=\{(x,p):p=0\}
\end{equation}
is invariant, and the restriction of \eqref{eq:ham-x}--\eqref{eq:ham-p} to $\Z$ is exactly the deterministic system $\dot x=b(x)$.  Thus the deterministic dynamics is embedded unchanged inside the larger Freidlin--Wentzell Hamiltonian system.  Trajectories with $p\ne0$ describe departures from deterministic motion and are the candidates for minimum-action fluctuation paths.  The connection with the free-time quasipotential will be made only after the invariant manifold formed by these extremals has been constructed in \Cref{sec:lagrangian}.  The finite-time boundary-reaching problem of \Cref{sec:switching-law,sec:solvent-application} uses the same control action \eqref{eq:control-action}, but it is not an equilibrium quasipotential problem.

\subsection{Normal-hyperbolicity assumptions}

The deterministic hypothesis needed below is normal hyperbolicity of $N$.  We recall the form needed here; for the general theory see Fenichel \cite{Fenichel1971,Fenichel1974}, Hirsch, Pugh, and Shub \cite{HirschPughShub1977}, and the geometric treatments \cite{Wiggins1994,Eldering2013}.  At each point $x\in N$, the tangent space of the ambient state space splits into directions that contract toward $N$, directions tangent to $N$, and directions that expand away from $N$:
\begin{equation}
T_xM=E_x^s\oplus T_xN\oplus E_x^u.
\end{equation}
As $x$ varies over $N$, each family of subspaces varies continuously.  The splitting is \emph{invariant} if the derivative of the deterministic flow carries each family into the corresponding family; for example,
\begin{equation}
D\phi_t(x)E_x^s=E_{\phi_t(x)}^s,
\end{equation}
and similarly for $TN$ and $E^u$.  We write the resulting family compactly as
\begin{equation}
T_NM=E^s\oplus TN\oplus E^u.
\label{eq:det-splitting}
\end{equation}
For notational simplicity, we use uniform exponential bounds.

\begin{assumption}[Deterministic normal hyperbolicity]
\label{ass:NHIM}
There are constants $C\ge1$ and $0\le\lambda_{\parallel}<\lambda_{\perp}$ such that, for $t\ge0$,
\begin{align}
\|D\phit v^s\|&\le Ce^{-\lambda_{\perp} t}\|v^s\|,
&v^s&\in E^s,
\label{eq:rate-s}\\
\|D\phi_{-t}v^u\|&\le Ce^{-\lambda_{\perp} t}\|v^u\|,
&v^u&\in E^u,
\label{eq:rate-u}\\
\|D\phi_{\pm t}v^c\|&\le Ce^{\lambda_{\parallel} t}\|v^c\|,
&v^c&\in TN.
\label{eq:rate-c}
\end{align}
\end{assumption}

The inequality $\lambda_{\perp}>\lambda_{\parallel}$ expresses the essential rate separation: contraction and expansion normal to $N$ dominate the possible growth tangent to $N$.  Stronger versions of this separation will be stated when additional smoothness or Lagrangian properties require them.

\section{The Freidlin--Wentzell geometry near \texorpdfstring{$N_0$}{N0}}
\label{sec:carrier}

We now ask how the deterministic stable, tangent, and unstable directions along $N$ appear in the Freidlin--Wentzell Hamiltonian system near the zero-momentum copy
\begin{equation}
N_0=N\times\{0\}\subset T^*M.
\end{equation}
The ingredients used in the calculation are standard: variational equations for Hamiltonian systems, the inverse-transpose evolution of covectors, invariant splittings associated with normal hyperbolicity, and center-manifold theory \cite{MarsdenRatiu1999,Fenichel1971,HirschPughShub1977,ChowLiuYi2000Smooth,ChowLiuYi2000Sets,Mielke1991}.  Exact Lagrangian stable and unstable manifolds also occur naturally in Hamiltonian dynamics \cite{RudnevWiggins1999}.  The purpose of this section is to show how these ingredients fit together for the specific Hamiltonian generated by the Freidlin--Wentzell variational problem.

Fix a deterministic orbit $x(t)=\phit(x_0)\in N$.  A perturbation of $(x(t),0)$ in $T^*M$ has two $n$-component parts: a displacement $\xi(t)$ of the original state and a perturbation $\eta(t)$ of the conjugate covector.  Set
\begin{equation}
A(t)=Db(x(t)),
\qquad
a(t)=a(x(t)).
\end{equation}
Linearizing \eqref{eq:ham-x} at $p=0$ gives
\begin{equation}
\dot\xi=A(t)\xi+a(t)\eta.
\label{eq:variational-x}
\end{equation}
Terms involving derivatives of $a$ are multiplied by the base value $p=0$ and therefore disappear.  Linearizing \eqref{eq:ham-p} at $p=0$ gives
\begin{equation}
\dot\eta=-A(t)^{\cT}\eta.
\label{eq:variational-p}
\end{equation}
Together,
\begin{equation}
\begin{pmatrix}\dot\xi\\ \dot\eta\end{pmatrix}
=
\begin{pmatrix}
A(t)&a(t)\\
0&-A(t)^{\cT}
\end{pmatrix}
\begin{pmatrix}\xi\\ \eta\end{pmatrix}.
\label{eq:variational-block}
\end{equation}
Here $\xi,\eta\in\R^n$, the two diagonal blocks and $a(t)$ are $n\times n$, and the full coefficient matrix is $2n\times2n$.

Let $F_t$ be the fundamental matrix of the deterministic variational equation,
\begin{equation}
\dot F_t=A(t)F_t,
\qquad F_0=I,
\qquad F_t=D\phit(x_0).
\label{eq:F-def}
\end{equation}
Equation~\eqref{eq:variational-p} is independent of $\xi$, and its solution is
\begin{equation}
\eta(t)=F_t^{-\cT}\eta_0.
\label{eq:eta-solution}
\end{equation}
Once $\eta(t)$ is known, it acts as an inhomogeneous forcing term in \eqref{eq:variational-x}.  Variation of constants gives
\begin{equation}
\xi(t)=F_t\xi_0+
F_t\int_0^tF_s^{-1}a(s)F_s^{-\cT}\eta_0\,\dd s.
\label{eq:xi-solution}
\end{equation}
Define the $n\times n$ matrix
\begin{equation}
K_t=
F_t\int_0^tF_s^{-1}a(s)F_s^{-\cT}\,\dd s.
\label{eq:K-def}
\end{equation}
Then the fundamental matrix $\Psi_t$ of the complete $2n$-dimensional linearized Freidlin--Wentzell system is
\begin{equation}
\Psi_t=
\begin{pmatrix}
F_t&K_t\\0&F_t^{-\cT}
\end{pmatrix},
\qquad
\begin{pmatrix}\xi(t)\\\eta(t)\end{pmatrix}
=\Psi_t\begin{pmatrix}\xi_0\\\eta_0\end{pmatrix}.
\label{eq:Psi}
\end{equation}
This triangular form is the key elementary observation.  The state perturbation evolves by the deterministic variational flow and is forced by $\eta$, whereas the covector perturbation evolves independently by the inverse transpose of that flow.

The appearance of $F_t^{-\cT}$ has a simple interpretation.  If $v$ is a deterministic tangent vector and $\eta$ is a covector, then
\begin{equation}
\bigl(F_t^{-\cT}\eta\bigr)\bigl(F_tv\bigr)=\eta(v).
\label{eq:dual-pairing}
\end{equation}
Thus their pairing is unchanged by the simultaneous tangent and covector evolutions.  If $F_tv$ grows like $e^{\lambda t}$ in some deterministic direction, a covector paired with that direction evolves with the opposite exponential rate, $e^{-\lambda t}$.  This is the elementary reason that deterministic expanding directions give contracting covector directions and deterministic contracting directions give expanding covector directions.

Because \eqref{eq:variational-block} is the variational equation of a Hamiltonian flow, $\Psi_t$ is symplectic \cite{MarsdenRatiu1999}.  This can also be checked directly from \eqref{eq:Psi}: since $a(s)$ is symmetric,
\begin{equation}
F_t^{-1}K_t
=\int_0^tF_s^{-1}a(s)F_s^{-\cT}\,\dd s
\label{eq:symmetric-FK}
\end{equation}
is symmetric, which is the required condition for the triangular matrix in \eqref{eq:Psi} to preserve the canonical symplectic form.

One technical point about rates will be used repeatedly below.  The off-diagonal block $K_t$ may introduce polynomial factors multiplying the exponential behavior inherited from $F_t$.  Such factors can always be absorbed into an arbitrarily small exponential loss: for every $\delta>0$ and every fixed integer $j\ge0$,
\begin{equation}
t^j e^{\lambda t}\le C_{\delta,j}e^{(\lambda+\delta)t},
\qquad t\ge0.
\label{eq:polynomial-loss}
\end{equation}
We will use this elementary estimate after the Hamiltonian stable, center, and unstable subspaces have been constructed.  The small loss $\delta$ is standard in invariant-manifold rate estimates \cite{Fenichel1974,HirschPughShub1977,ChowLiuYi2000Smooth}.

Before stating the splitting theorem, it is useful to see the dimension count.  Write
\begin{equation}
s=\dim E^s,
\qquad u=\dim E^u,
\qquad s+k+u=n.
\end{equation}
The deterministic stable directions contribute $s$ stable state directions.  The $u$ deterministic unstable directions contribute $u$ stable covector directions because of the rate reversal in \eqref{eq:dual-pairing}.  Hence the Hamiltonian stable subspace has dimension $s+u=n-k$.  The $k$ tangent directions to $N$ are accompanied by $k$ covector directions associated with $TN$, giving $2k$ center directions.  Reversing the argument gives another $n-k$ unstable directions.

The precise bookkeeping uses covectors that vanish on complementary deterministic directions.  For a subspace $V$, such covectors form its annihilator.  In the present splitting we use the notation
\begin{align}
(E^u)^*&=\{\eta:\eta|_{E^s\oplus TN}=0\},
&\dim(E^u)^*&=u,\\
T^*N&=\{\eta:\eta|_{E^s\oplus E^u}=0\},
&\dim T^*N&=k,\\
(E^s)^*&=\{\eta:\eta|_{TN\oplus E^u}=0\},
&\dim(E^s)^*&=s.
\label{eq:annihilator-splitting}
\end{align}
The notation simply records which deterministic directions a covector measures.  It is useful because $\eta$ is intrinsically a covector, not another state-space vector.

We will also use two standard symplectic terms.  A subspace is \emph{isotropic} if the symplectic form vanishes on it.  A half-dimensional isotropic subspace is \emph{Lagrangian}.  A subspace is \emph{symplectic} if the restriction of $\omega$ to it is nondegenerate \cite{MarsdenRatiu1999}.

The following theorem constructs the three invariant groups of directions.  We distinguish them from the deterministic spaces $E^s$, $TN$, and $E^u$ by writing
\[
E_H^s,\qquad E_H^c,\qquad E_H^u,
\]
where the subscript $H$ means that these are the stable, center, and unstable subspaces of the \emph{Hamiltonian} variational equation.

\begin{theorem}[Hamiltonian exponential trichotomy]
\label{thm:trichotomy}
Under Assumption~\ref{ass:NHIM}, the Hamiltonian variational flow along $N_0$ admits an invariant splitting
\begin{equation}
T_{N_0}(T^*M)=E_H^s\oplus E_H^c\oplus E_H^u
\label{eq:H-splitting}
\end{equation}
with
\begin{equation}
\dim E_H^s=n-k,
\qquad
\dim E_H^c=2k,
\qquad
\dim E_H^u=n-k.
\label{eq:H-dimensions}
\end{equation}
The center bundle $E_H^c$ is symplectic, and $TN_0\subset E_H^c$ is Lagrangian.
\end{theorem}

\begin{proof}
The inverse-transpose equation \eqref{eq:variational-p} preserves the three covector families in \eqref{eq:annihilator-splitting}.  By \eqref{eq:dual-pairing}, covectors in $(E^u)^*$ contract forward at the rates opposite to the deterministic expansion in $E^u$, covectors in $T^*N$ have the center rates opposite to those in $TN$, and covectors in $(E^s)^*$ expand forward at the rates opposite to the deterministic contraction in $E^s$.

If the off-diagonal term $a(t)\eta$ were absent, the Hamiltonian stable, center, and unstable spaces would therefore be the direct sums
\begin{equation}
E^s\oplus(E^u)^*,
\qquad
TN\oplus T^*N,
\qquad
E^u\oplus(E^s)^*.
\label{eq:exact-graph-bases}
\end{equation}
The forcing term $a(t)\eta$ means that these direct sums are not, in general, the literal invariant subspaces.  Instead, the required state components are uniquely corrected so that the complete solution has the desired stable, center, or unstable behavior.  For example, if $\eta_c(t)$ is a center covector solution, the normal state components that remain at the center rate are
\begin{align}
\xi_s(t)&=\int_{-\infty}^{t}F^s(t,\tau)P^s(\tau)a(\tau)\eta_c(\tau)\,\dd\tau,\\
\xi_u(t)&=-\int_t^\infty F^u(t,\tau)P^u(\tau)a(\tau)\eta_c(\tau)\,\dd\tau.
\label{eq:center-normal-corrections}
\end{align}
These formulas are the usual variation-of-constants formulas with the integration limits chosen so that the unwanted exponentially growing solution is absent.  Here $P^s(\tau)$ and $P^u(\tau)$ are the projections onto the deterministic stable and unstable subspaces, while $F^s(t,\tau)$ and $F^u(t,\tau)$ denote the corresponding restricted evolution operators.  The deterministic center-rate bounds, the inverse-transpose rate reversal, and the polynomial-loss estimate \eqref{eq:polynomial-loss} imply that the center data grow at most like $e^{(\lambda_{\parallel}+\delta)|t|}$ for any fixed $\delta>0$.  Together with the normal rate gap, this implies convergence of the improper integrals.  The stable and unstable cases are analogous.  Thus the actual invariant spaces are graphs over the three direct sums in \eqref{eq:exact-graph-bases}.  Their dimensions are unchanged, giving \eqref{eq:H-dimensions}.

It remains to identify the symplectic character of the splitting.  If $v,w\in E_H^s$, symplectic invariance gives
\begin{equation}
\omega(v,w)=\omega(\Psi_tv,\Psi_tw).
\end{equation}
Both vectors on the right tend exponentially to zero as $t\to+\infty$, so
\begin{equation}
|\omega(\Psi_tv,\Psi_tw)|
\le C\|\Psi_tv\|\,\|\Psi_tw\|\longrightarrow0.
\end{equation}
The left-hand side is independent of $t$, and hence $\omega(v,w)=0$.  Thus $E_H^s$ is isotropic.  The same argument in backward time shows that $E_H^u$ is isotropic.

The same construction gives the center-rate estimate
\begin{equation}
\|\Psi_{\pm t}|_{E_H^c}\|\le C_\delta e^{(\lambda_{\parallel}+\delta)t},
\qquad t\ge0,
\label{eq:center-growth-delta}
\end{equation}
for every sufficiently small $\delta>0$.  If $v_s\in E_H^s$ and $v_c\in E_H^c$, this estimate gives
\begin{equation}
|\omega(v_s,v_c)|
=|\omega(\Psi_tv_s,\Psi_tv_c)|
\le C_\delta e^{-[\lambda_{\perp}-(\lambda_{\parallel}+\delta)]t}\longrightarrow0.
\end{equation}
Hence $\omega(E_H^s,E_H^c)=0$; backward time gives $\omega(E_H^u,E_H^c)=0$.  Therefore
\begin{equation}
E_H^c=(E_H^s\oplus E_H^u)^\omega.
\end{equation}
If the restriction of $\omega$ to $E_H^c$ were degenerate, a nonzero vector in its kernel would be symplectically orthogonal to all three summands in \eqref{eq:H-splitting}, contradicting nondegeneracy of the ambient symplectic form.  Hence $E_H^c$ is symplectic.  Finally, $TN_0$ lies in the zero section, so $\omega$ vanishes on $TN_0$.  Since
\begin{equation}
\dim TN_0=k=\frac12\dim E_H^c,
\end{equation}
$TN_0$ is Lagrangian inside the symplectic center space $E_H^c$.
\end{proof}

The theorem has a simple geometric interpretation.  The deterministic NHIM contributes $k$ tangent directions.  The Freidlin--Wentzell Hamiltonian system supplies $k$ conjugate covector directions paired with them by the symplectic form.  A point of a cotangent bundle $T^*N$ consists of exactly these two pieces of information: a point of $N$ and a covector attached to that point.  This is why a cotangent bundle appears in the local center geometry.

More precisely, after quotienting out the tangent directions $TN_0$, the remaining center directions are canonically identified with covectors on $N$ by
\begin{equation}
E_H^c/TN_0\longrightarrow T^*N,
\qquad
[v]\longmapsto \omega(v,\cdot)|_{TN_0}.
\label{eq:center-quotient}
\end{equation}
The quotient formula is the coordinate-independent statement.  The preceding paragraph gives the simpler picture: the $2k$ center directions consist locally of the $k$ directions along $N$ and their $k$ conjugate covector directions.

Theorem~\ref{thm:trichotomy} has so far described only tangent directions along $N_0$.  It has not yet produced a $2k$-dimensional nonlinear invariant manifold containing $N_0$.  Center-manifold theory provides that next step.  The required rate condition says that the normal contraction and expansion dominate the center growth strongly enough to obtain the desired smoothness.

\begin{theorem}[Local symplectic FW center manifold]
\label{thm:carrier}
Let $N$ satisfy Assumption~\ref{ass:NHIM}, and suppose $b,a\in C^r$.  If $1\le\ell\le r$ and, for some $\delta>0$, the rate-separation condition
\begin{equation}
\ell(\lambda_{\parallel}+\delta)<\lambda_{\perp}
\label{eq:center-rate-separation}
\end{equation}
holds, then there exists a $C^\ell$ locally invariant manifold
\begin{equation}
\C_{\mathrm{FW}}\subset T^*M,
\qquad
\dim \C_{\mathrm{FW}}=2k,
\label{eq:carrier}
\end{equation}
such that
\begin{equation}
N_0\subset\C_{\mathrm{FW}},
\qquad
T_{N_0}\C_{\mathrm{FW}}=E_H^c.
\label{eq:carrier-tangent}
\end{equation}
After restricting to a sufficiently small neighborhood of $N_0$, $\C_{\mathrm{FW}}$ is symplectic.  The pair $(\C_{\mathrm{FW}},N_0)$ is locally symplectomorphic to $(T^*N,N)$.
\end{theorem}

\begin{proof}
The compact invariant manifold $N_0$ and the uniform stable-center-unstable splitting in \Cref{thm:trichotomy} satisfy the hypotheses of center-manifold theory for smooth invariant manifolds \cite{ChowLiuYi2000Smooth}; one may equivalently use the invariant-set formulation \cite{ChowLiuYi2000Sets}.  The rate separation \eqref{eq:center-rate-separation} ensures that the normal contraction and expansion dominate the possible center growth.  The resulting $C^\ell$ invariant manifold contains $N_0$ and is tangent to $E_H^c$ along $N_0$.

Since $\omega|_{E_H^c}$ is nondegenerate and nondegeneracy is an open condition, the restriction of $\omega$ to $\C_{\mathrm{FW}}$ remains nondegenerate after the neighborhood is made sufficiently small.  Hence $\C_{\mathrm{FW}}$ is symplectic.  The Hamiltonian vector field is tangent to $\C_{\mathrm{FW}}$, so the restricted dynamics is Hamiltonian with Hamiltonian $\Hs|_{\C_{\mathrm{FW}}}$.  The submanifold $N_0$ is Lagrangian in $\C_{\mathrm{FW}}$; Weinstein's Lagrangian-neighborhood theorem \cite{Weinstein1971} then identifies a neighborhood of $N_0$ in $\C_{\mathrm{FW}}$ symplectically with a neighborhood of the zero section in $T^*N$.
\end{proof}

\begin{remark}[What is and is not canonical]
The manifold $\C_{\mathrm{FW}}$ is constructed only in a neighborhood of $N_0$, and center manifolds need not be unique away from $N_0$.  The invariant information established above is the $2k$-dimensional center geometry along $N_0$ and its local symplectic cotangent structure.  The next section introduces a different invariant manifold that is defined directly by the dynamics and does not depend on choosing a center manifold.
\end{remark}

\section{The invariant manifold of backward-asymptotic fluctuation extremals}
\label{sec:lagrangian}

The local center manifold describes the Hamiltonian geometry near $N_0$, but the large-deviation problem requires a different invariant object.  Since $N_0$ is invariant, no Hamiltonian trajectory starting at a point of $N_0$ at finite time can leave it.  Instead, for each $z_0\in N_0$, consider the points whose trajectories approach the trajectory through $z_0$ exponentially as $t\to-\infty$ at the strong normal rate.  These trajectories are backward asymptotic to $N_0$.  The corresponding points form the strong-unstable fiber through $z_0$, and taking the union of the fibers over $N_0$ gives
\begin{equation}
\Lagr_N^+=W^{uu}(N_0).
\label{eq:Lambda-def}
\end{equation}
The double $u$ emphasizes that this is the \emph{strong} unstable manifold.  At each point of the $k$-dimensional set $N_0$ there are $n-k$ strong-unstable directions, so the resulting manifold has dimension $k+(n-k)=n$.

The theorem below describes $W^{uu}(N_0)$ as a submanifold of $T^*M$.  The word \emph{immersed} means that the map into $T^*M$ has injective derivative at every point, so locally its image is a smooth submanifold; globally, distinct parts of the abstract manifold are not required to have disjoint images.

\begin{theorem}[Exact Lagrangian fluctuation manifold]
\label{thm:lagrangian}
Assume that Assumption~\ref{ass:NHIM} holds and choose $\delta>0$ so that
\begin{equation}
2(\lambda_{\parallel}+\delta)<\lambda_{\perp}.
\label{eq:two-bunching}
\end{equation}
Then $W^{uu}(N_0)$ immerses into $T^*M$ as an $n$-dimensional exact Lagrangian invariant manifold contained in $\Hs^{-1}(0)$.  If
\[
\iota:W^{uu}(N_0)\longrightarrow T^*M
\]
denotes the immersion, then
\begin{equation}
\iota^*\omega=0,
\qquad
\iota^*\vartheta=\dd S
\label{eq:exact-lag}
\end{equation}
for a scalar action function $S$.
\end{theorem}

\begin{proof}
The strong-unstable manifold theorem for a compact partially hyperbolic invariant manifold supplies a $C^1$ immersed strong-unstable manifold tangent along $N_0$ to
\begin{equation}
TW^{uu}(N_0)=TN_0\oplus E_H^u
\qquad\text{on }N_0
\label{eq:Wu-tangent}
\end{equation}
\cite{HirschPughShub1977,Fenichel1974}.  The first summand has dimension $k$ and the second has dimension $n-k$, hence
\[
\dim W^{uu}(N_0)=n.
\]

For $z\in W^{uu}(N_0)$, conservation of the Freidlin--Wentzell Hamiltonian and convergence to $N_0$ in backward time give
\[
\Hs(z)
=
\lim_{t\to-\infty}\Hs(\Phi_tz)
=0,
\]
because $\Hs=0$ on $N_0$.  Hence
\[
W^{uu}(N_0)\subset\Hs^{-1}(0).
\]

We now use the canonical one-form $\vartheta=p\cdot\dd x$.  Let $X_{\Hs}$ denote the Hamiltonian vector field.  With our convention $\omega=-\dd\vartheta$, Cartan's formula gives
\begin{equation}
\mathcal L_{X_{\Hs}}\vartheta
=
\dd\!\left[\vartheta(X_{\Hs})-\Hs\right]
\label{eq:cartan-action}
\end{equation}
\cite{MarsdenRatiu1999}.  Integrating \eqref{eq:cartan-action} along the Hamiltonian flow from $-T$ to $0$ yields
\begin{equation}
\vartheta-\Phi_{-T}^*\vartheta
=
\dd S_T,
\qquad
S_T(z)
=
\int_{-T}^{0}
\left[\vartheta(X_{\Hs})-\Hs\right](\Phi_tz)\,\dd t .
\label{eq:finite-action-primitive}
\end{equation}

We next let $T\to\infty$ on $W^{uu}(N_0)$.  By definition of the strong-unstable fibers, the momentum component of $\Phi_{-T}z$ tends to zero at the strong normal rate.  A tangent vector to $W^{uu}(N_0)$ transported backward is the sum of a component tangent to the base orbit in $N_0$, whose growth is bounded by $e^{(\lambda_\parallel+\delta)T}$, and a strong-unstable component, which contracts.  The invariant-fiber estimates therefore give, for $v\in T_zW^{uu}(N_0)$,
\[
\left|
(\Phi_{-T}^*\vartheta)_z(v)
\right|
=
\left|
\vartheta_{\Phi_{-T}z}(D\Phi_{-T}v)
\right|
\le
C_v e^{-[\lambda_\perp-(\lambda_\parallel+\delta)]T}
\longrightarrow0 .
\]
The stronger hypothesis \eqref{eq:two-bunching} in particular guarantees the positive exponent needed here.  These invariant-fiber estimates are locally uniform on compact subsets of $W^{uu}(N_0)$, so $\Phi_{-T}^*\vartheta\to0$ locally uniformly there.

The integral in \eqref{eq:finite-action-primitive} also converges locally uniformly on compact subsets.  Since $\dd S_T=\vartheta-\Phi_{-T}^*\vartheta$, both $S_T$ and $\dd S_T$ converge locally uniformly on compact subsets; hence the limiting function $S$ is $C^1$ and satisfies $\iota^*\vartheta=\dd S$.  Indeed, Hamilton's equation gives $\dot x=b+ap$, and therefore
\begin{align}
\vartheta(X_{\Hs})-\Hs
&=
p\cdot\dot x
-
\left(p\cdot b+\frac12p^{\cT}ap\right)\\
&=
\frac12p^{\cT}ap
=
\frac12|B^{\cT}p|^2.
\label{eq:action-density-identity}
\end{align}
Along a strong-unstable orbit, $p(t)$ tends to zero exponentially as $t\to-\infty$, so the right-hand side is integrable.  The limiting function is therefore
\begin{equation}
S(z)
=
\int_{-\infty}^{0}
\left[\vartheta(X_{\Hs})-\Hs\right](\Phi_tz)\,\dd t.
\label{eq:S-action-general}
\end{equation}
Taking an exterior derivative of $\iota^*\vartheta=\dd S$ and using $\omega=-\dd\vartheta$ gives
\[
\iota^*\omega=0.
\]
Since $W^{uu}(N_0)$ has dimension $n$, one half of the dimension of $T^*M$, it is Lagrangian.  Thus it is exact Lagrangian.

On $W^{uu}(N_0)$ we already know that $\Hs=0$, so \eqref{eq:S-action-general} becomes
\begin{equation}
S(z)=\int_{-\infty}^{0}p(t)\cdot\dot x(t)\,\dd t.
\label{eq:S-action}
\end{equation}
Moreover, the maximizing control in the Freidlin--Wentzell variational principle is $u_*=B^{\cT}p$ by \eqref{eq:optimal-control-p}.  Equation~\eqref{eq:action-density-identity} therefore shows directly that
\begin{equation}
S(z)
=
\frac12\int_{-\infty}^{0}|u_*(t)|^2\,\dd t.
\label{eq:S-control-action}
\end{equation}
Thus the generating function on the fluctuation manifold is exactly the accumulated Freidlin--Wentzell control cost.  In particular $S\ge0$, and $S=0$ on $N_0$.
\end{proof}

Exact Lagrangian stable and unstable manifolds are familiar in Hamiltonian dynamics \cite{RudnevWiggins1999}.  Their significance here is that the scalar generating function is precisely the action accumulated along the Hamiltonian fluctuation extremals in $W^{uu}(N_0)$.  The manifold need not project to the original state space as a single graph.  This point is classical in weak-noise asymptotics: folds of the Lagrangian projection generate caustics, and competition between action branches produces switching sets \cite{GrahamTel1985,MaierStein1993PRL,MaierStein1996,SmelyanskiyDykmanMaier1997}.

We can now connect this geometry with the quasipotential.  In a \emph{free-time} problem the arrival time is not prescribed; one minimizes over both the path and the travel time.  When the quasipotential from $N$ is the appropriate object, define
\begin{equation}
V_N(x)=\inf_{T>0}\inf_{\substack{\gamma(-T)\in N\\ \gamma(0)=x}}
I_{[-T,0]}[\gamma].
\label{eq:quasipotential-N}
\end{equation}
For an autonomous free-time problem, the relevant Hamiltonian extremals lie on $\Hs=0$ \cite{FreidlinWentzell2012}.  Theorem~\ref{thm:lagrangian} identifies the invariant manifold in which the extremals that approach $N_0$ in backward time at the strong normal rate lie.

Let
\begin{equation}
\pi:T^*M\longrightarrow M,
\qquad
\pi(x,p)=x,
\label{eq:cotangent-projection}
\end{equation}
be the natural projection from the extended Hamiltonian phase space to the original state space.  Thus $\pi$ retains the state $x$ and discards the auxiliary conjugate variable $p$.  If $\pi$ is locally one-to-one on a branch of $W^{uu}(N_0)$, that branch can be written as $p=p(x)$ and exactness gives
\begin{equation}
\dd V_N=p\cdot\dd x,
\qquad
\Hs(x,\nabla V_N(x))=0.
\label{eq:HJ}
\end{equation}
A fold caustic occurs when this projection ceases to be locally one-to-one,
\begin{equation}
\rank D\pi|_{\Lagr_N^+}<n.
\label{eq:caustic-condition}
\end{equation}
Then distinct points $(x,p_j)$ of the fluctuation manifold project to the same physical state $x$, corresponding to distinct extremal fluctuation paths with different conjugate momenta and, in general, different action values.  The quasipotential selects the smallest action.  A switching set occurs where two minimizing branches have equal action and exchange which branch is selected.  Thus a caustic records multiplicity of extremals, whereas a switching set records competition between minimizers \cite{MaierStein1993PRL,MaierStein1996,MaierStein1997,SmelyanskiyDykmanMaier1997}.

\section{Normal hyperbolicity need not pass to the zero-energy section}
\label{sec:counterexample}

The local center manifold $\C_{\mathrm{FW}}$ is invariant under the Freidlin--Wentzell Hamiltonian flow.  Consequently the Hamiltonian vector field is tangent to $\C_{\mathrm{FW}}$, and the motion restricted to this manifold is itself Hamiltonian with Hamiltonian function $\Hs|_{\C_{\mathrm{FW}}}$.  The zero-energy trajectories within the center manifold therefore lie in
\begin{equation}
\Sigma_{\mathrm{FW}}^{\mathrm{loc}}
=
\C_{\mathrm{FW}}\cap\Hs^{-1}(0).
\label{eq:local-fw-nhim}
\end{equation}

For this intersection to be a smooth hypersurface, zero must be a \emph{regular value} of the restricted Hamiltonian: its differential must not vanish on the zero-energy set.  When that condition holds, the regular-value theorem gives a $(2k-1)$-dimensional smooth level set inside the $2k$-dimensional center manifold.  At a point $(x,0)\in N_0$,
\begin{equation}
\dd\Hs(\xi,\eta)=\eta\cdot b(x).
\label{eq:dH-N0}
\end{equation}
Thus, when $b(x)\ne0$ on $N$, a covector $\eta$ can be chosen with $\eta\cdot b(x)\ne0$, and the zero-energy level is locally regular near $N_0$.

The natural question is whether normal hyperbolicity of the original deterministic manifold $N$ automatically makes this zero-energy section normally hyperbolic within the ambient zero-energy hypersurface $\Hs^{-1}(0)$.  It does not.  The reason is that the deterministic normal rates are measured along trajectories with $p=0$, whereas a point of the zero-energy section can have $p\ne0$.  The conjugate momentum then changes the evolution of the original state variables, so the Hamiltonian trajectory may sample the transverse stability coefficient in a way that no deterministic trajectory does.  The following example is constructed solely to make that mechanism explicit.

\begin{proposition}[Normal hyperbolicity need not be inherited within the zero-energy hypersurface]
\label{prop:counterexample}
There exists a compact normally attracting deterministic invariant torus $N$ and a positive-definite diffusion tensor for which an exactly invariant symplectic center manifold has a compact regular zero-energy shell that is not normally hyperbolic as an invariant manifold of the zero-energy hypersurface $\Hs^{-1}(0)$.
\end{proposition}

\begin{proof}
Let $M=\T^2\times\R$ with coordinates $(\theta_1,\theta_2,z)$, and consider
\begin{equation}
\dot\theta_1=1,
\qquad
\dot\theta_2=0,
\qquad
\dot z=(-1+2\cos\theta_1)z.
\label{eq:counter-drift}
\end{equation}
The torus $N=\T^2\times\{0\}$ is invariant.  Along a deterministic orbit, $\theta_1(t)=\theta_1(0)+t$ and
\[
z(t)=z(0)\exp\!\left[-t+2\sin(\theta_1(0)+t)-2\sin\theta_1(0)\right].
\]
Since the sine terms remain bounded,
\[
|z(t)|\le e^4e^{-t}|z(0)|.
\]
Thus perturbations normal to the torus contract exponentially at asymptotic rate one, while the two tangent directions have zero exponential rate.  Hence $N$ is normally attracting.

Choose the constant positive-definite diffusion tensor $a=\diag(1,1,q)$, $q>0$.  With conjugate variables $(\eta_1,\eta_2,\zeta)$, the Freidlin--Wentzell Hamiltonian is
\begin{equation}
\Hs
=
\eta_1+(-1+2\cos\theta_1)z\zeta
+\frac12(\eta_1^2+\eta_2^2+q\zeta^2).
\label{eq:counter-H}
\end{equation}
The four-dimensional manifold $\C_0=\{z=\zeta=0\}$ is exactly invariant and symplectic.  On it the Hamiltonian reduces to
\[
h(\eta_1,\eta_2)=\eta_1+\frac12(\eta_1^2+\eta_2^2),
\]
so the zero-energy shell is
\begin{equation}
(\eta_1+1)^2+\eta_2^2=1,
\qquad
\Sigma_{0,0}\simeq\T^2\times S^1.
\label{eq:counter-shell}
\end{equation}
This shell is compact.  It is also regular because the gradient of $h$ with respect to $(\eta_1,\eta_2)$ is $(1+\eta_1,\eta_2)$, which cannot vanish on the circle \eqref{eq:counter-shell}.

Now select the point on the momentum circle with $\eta_1=-1$ and $\eta_2=1$.  The Hamiltonian equations restricted to $\C_0$ give
\[
\dot\theta_1=1+\eta_1=0,
\qquad
\dot\theta_2=\eta_2=1.
\]
Thus the Hamiltonian trajectory keeps $\theta_1$ constant.  Choose the constant value $\theta_1=\pi/3$.  At this value $-1+2\cos\theta_1=0$.  The transverse linearized equations along the selected Hamiltonian trajectory are therefore
\begin{equation}
\dot z=q\zeta,
\qquad
\dot\zeta=0,
\label{eq:nilpotent-transverse}
\end{equation}
or
\[
\frac{\dd}{\dd t}\begin{pmatrix}z\\ \zeta\end{pmatrix}
=
\begin{pmatrix}0&q\\0&0\end{pmatrix}
\begin{pmatrix}z\\ \zeta\end{pmatrix}.
\]
The transverse matrix is nilpotent: its square is zero.  Consequently
\[
e^{tA_\perp}=I+tA_\perp
=
\begin{pmatrix}1&qt\\0&1\end{pmatrix}.
\]
Transverse perturbations therefore grow at most linearly rather than exponentially, and both transverse Lyapunov exponents are zero.  There is no exponential transverse contraction or expansion along this trajectory.  Because the $(z,\zeta)$ directions are tangent to $\Hs^{-1}(0)$ along the shell and normal to the shell within that hypersurface, normal hyperbolicity of the zero-energy shell fails there.
\end{proof}

The failure is not caused by a singular energy level, noncompactness, or loss of invariance: the zero-energy shell is compact and regular, and $\C_0$ is exactly invariant.  What changes is the trajectory along which the transverse rate is sampled.  In the deterministic system, $\dot\theta_1=1$, so the orbit continually passes through all values of $\theta_1$.  The instantaneous coefficient $-1+2\cos\theta_1$ may vanish at isolated phases, but its time average is $-1$, producing net exponential contraction.  In the Hamiltonian system, the nonzero tangential momentum $\eta_1=-1$ changes the equation to $\dot\theta_1=0$.  The trajectory can therefore remain permanently at the value $\theta_1=\pi/3$, where the instantaneous transverse contraction is zero.

This counterexample does not invalidate the local results of the preceding sections.  Those results concern the Hamiltonian variational dynamics \emph{along $N_0$}, where $p=0$, and use the deterministic normal rate gap there.  In the present example $\lambda_\parallel=0$ and $\lambda_\perp=1$, so the stronger condition $2(\lambda_\parallel+\delta)<\lambda_\perp$ also holds for sufficiently small $\delta>0$.  The local symplectic center manifold near $N_0$ and the exact Lagrangian strong-unstable manifold of fluctuation extremals therefore still exist.  What fails is the stronger assertion that the whole zero-energy section is normally hyperbolic within $\Hs^{-1}(0)$, or equivalently that its nonzero-momentum trajectories must inherit the deterministic transverse rates.  Any theorem making such a claim requires additional information about transverse rates along the Hamiltonian trajectories on that section.

\section{A solvable example of the fluctuation manifold}
\label{sec:exact-torus}

The preceding results are general.  Before turning to the different role of invariant manifolds in the reaction problem, it is useful to see the fluctuation manifold and its action in one elementary model where every quantity can be written explicitly.  Let
\[
\theta\in\T^k,\qquad r\in\R^m,\qquad m=n-k,
\]
and consider the deterministic system
\begin{equation}
\dot\theta=\Omega,
\qquad
\dot r=-Ar,
\label{eq:linear-torus-drift}
\end{equation}
where all eigenvalues of $A$ have positive real parts.  The torus $N=\T^k\times\{0\}$ is normally attracting.  Let the diffusion tensor be block diagonal,
\begin{equation}
a=
\begin{pmatrix}
D&0\\0&Q
\end{pmatrix},
\qquad D>0,\quad Q>0.
\end{equation}
With conjugate variables $(\eta,\zeta)$, the Freidlin--Wentzell Hamiltonian is
\begin{equation}
H_0
=
\Omega\cdot\eta+\frac12\eta^{\cT}D\eta
-\zeta^{\cT}Ar+\frac12\zeta^{\cT}Q\zeta.
\label{eq:linear-torus-H}
\end{equation}

The center manifold is exactly $\C_0=\{r=\zeta=0\}\simeq T^*\T^k$.  The transverse equations are
\begin{equation}
\dot r=-Ar+Q\zeta,
\qquad
\dot\zeta=A^{\cT}\zeta.
\label{eq:linear-normal}
\end{equation}
Because $A$ is stable, the matrix
\begin{equation}
K=\int_0^\infty e^{-As}Qe^{-A^{\cT}s}\,\dd s
\label{eq:lyapunov-K}
\end{equation}
is well defined, symmetric, and positive definite; differentiating the integral gives the Lyapunov equation $AK+KA^{\cT}=Q$.  A solution of \eqref{eq:linear-normal} that approaches $r=\zeta=0$ in backward time must satisfy $r=K\zeta$.  Since $\dot\eta=0$, any Hamiltonian fluctuation extremal that is backward asymptotic to the zero-momentum torus must have $\eta=0$.  Hence
\begin{equation}
\Lagr_N^+=W^{uu}(N_0)=\{\eta=0,\ r=K\zeta\}.
\label{eq:linear-Lambda}
\end{equation}
Solving for the conjugate variable gives
\begin{equation}
\zeta=K^{-1}r=\nabla_rV(r),
\qquad
V(r)=\frac12r^{\cT}K^{-1}r.
\label{eq:linear-quasi}
\end{equation}

This example displays the general theorem in its simplest form.  The strong-unstable fluctuation manifold is a global graph over the normal state variable $r$, and its generating function is the quadratic action $V$.  Because the projection is globally one-to-one, this model has no caustics and no competing action branches.  Its role in the paper is only to make the abstract exact-Lagrangian construction concrete; no perturbation of this product model is needed for the later results.

\section{Quadratic action near a deterministic reactivity boundary}
\label{sec:switching-law}

We now turn to a different question from the one considered in Sections~\ref{sec:carrier}--\ref{sec:exact-torus}.  There the relevant Hamiltonian fluctuation extremals approached the zero-momentum copy $N_0$ of the deterministic NHIM in backward time at the strong normal rate.  Here the invariant structure plays a different role: it is part of deterministic reaction geometry, and its stable manifold separates trajectories with different reactive fates.

This geometric use of invariant manifolds has a substantial literature.  For two-degree-of-freedom and collinear models, the periodic-orbit formulation was developed well before the general NHIM theory: see Pechukas and McLafferty \cite{PechukasMcLafferty1973}, Pollak and Pechukas \cite{PollakPechukas1978}, and Pechukas and Pollak \cite{PechukasPollak1979}.  In that setting an unstable periodic orbit organizes a recrossing-free dividing surface and the trapped trajectories that form reactivity boundaries.  Wiggins, Wiesenfeld, Jaff\'e, and Uzer \cite{WigginsWiesenfeldJaffeUzer2001}, followed by Uzer, Jaff\'e, Palaci\'an, Yanguas, and Wiggins \cite{UzerEtAl2002}, generalized the phase-space construction to systems with three or more degrees of freedom, where the transition-state structure is a higher-dimensional NHIM and its stable and unstable manifolds form multidimensional separatrices.  Further geometric formulations are reviewed in \cite{WaalkensSchubertWiggins2008}.

The phase-space picture also extends to explicitly time-dependent Hamiltonian systems.  Mandell and Wiggins constructed time-dependent NHIMs, time-dependent stable and unstable manifolds, and time-dependent no-recrossing dividing surfaces for index-one reaction dynamics \cite{MandellWiggins2021}.  Cao and Wiggins treated time-dependent reaction dynamics associated with an index-two saddle and computed the corresponding time-dependent NHIM and its stable and unstable manifolds \cite{CaoWiggins2022}.  These deterministic nonautonomous constructions are useful to distinguish from the noise-realization-dependent geometry discussed below.

Kawai and Komatsuzaki showed that a reaction boundary separating reactive and nonreactive trajectories can remain dynamically meaningful even beyond regimes in which a conventional no-return transition-state surface exists \cite{KawaiKomatsuzaki2010}.  Nagahata, Teramoto, Li, Kawai, and Komatsuzaki developed this idea further under the name \emph{reactivity boundary}, including reactions involving higher-index and multiple saddles \cite{NagahataEtAl2013a,NagahataEtAl2013b}.  Later work used related reactivity-boundary ideas in driven and Langevin-coupled systems \cite{NagahataBorondoBenitoHernandez2020}.  We adopt that established terminology.

Stochastic reaction dynamics introduces other, genuinely different, notions of transition geometry.  Bartsch, Hernandez, Uzer, and collaborators constructed a noise-realization-dependent transition-state trajectory and a moving dividing surface designed to eliminate recrossings in fluctuating environments \cite{BartschHernandezUzer2005,BartschUzerHernandez2005,BartschUzerMoixHernandez2006,BartschEtAl2008,BartschRevueltaBenitoBorondo2012}.  Kawai and Komatsuzaki developed a related time-dependent normal-form approach for multidimensional Langevin equations, including a reaction coordinate that depends explicitly on the random force \cite{KawaiKomatsuzaki2009}.  These constructions are pathwise or noise-dependent in a way that the deterministic boundary below is not.  The stochastic transition state moves with the particular realization of the environmental forcing.  Transition-path theory takes a probabilistic viewpoint instead.  If $A$ and $B$ are reactant and product sets, the forward committor \cite{EVandenEijnden2006,MetznerSchuetteVandenEijnden2009} is defined by
\begin{equation}
q^+(x)=\Pr_x\{\tau_B<\tau_A\}.
\label{eq:committor}
\end{equation}
It is the probability of reaching $B$ before $A$.  Its level sets, especially $q^+=1/2$, provide surfaces of equal probabilistic commitment.  Transition-path theory uses forward and backward committors to describe the statistical ensemble of reactive trajectories, including reactive densities, currents, fluxes, rates, and dominant pathways \cite{EVandenEijnden2006,MetznerSchuetteVandenEijnden2009,EVandenEijnden2010}.

The deterministic reactivity boundary, the noise-realization-dependent stochastic transition state, and an isocommittor surface are therefore not interchangeable objects.  \Cref{thm:quadratic-switching} uses the first of these.  It asks a local Freidlin--Wentzell question: as a parameter brings an unforced trajectory toward a deterministic reactivity boundary, how rapidly does the minimum action required to reach that boundary go to zero?  This question complements the minimum-action methods used to compute rare-event paths \cite{ERenVandenEijnden2004,HeymannVandenEijnden2008}; it does not attempt to replace stochastic transition-state theory or transition-path theory.

Let a scalar parameter $\mu$ vary near a critical value $\mu_c$.  For each $\mu$, suppose that a smooth codimension-one deterministic reactivity boundary $\mathcal W_\mu$ separates two local reactive fates.  In the reaction application this boundary is a stable manifold associated with a transition-state invariant structure.  The system, the initial condition, and the boundary may all depend on $\mu$.  At $\mu=\mu_c$ the selected unforced trajectory reaches the corresponding boundary.  As $\mu$ changes, what matters is the \emph{relative} displacement of the trajectory and the moving boundary; no trajectory crosses an invariant stable manifold at fixed $\mu$.

Consider the controlled system
\begin{equation}
\dot x=f(x,\mu)+B(x,\mu)u(t),
\qquad
x(0)=x_0(\mu),
\label{eq:switch-control-system}
\end{equation}
where $x\in\R^d$, the scalar parameter $\mu$ varies near a critical value $\mu_c$, and $u\in L^2([0,T];\R^m)$.  The cost is
\begin{equation}
J_T[u]=\frac12\int_0^T |u(t)|^2\,\dd t.
\label{eq:switch-cost}
\end{equation}
Suppose that the deterministic system possesses a hyperbolic invariant set $\mathcal N_\mu$ whose stable manifold is the codimension-one deterministic reactivity boundary in a neighborhood of the critical trajectory.  We describe this boundary locally by a smooth signed function $\rho$:
\begin{equation}
\mathcal W_\mu=\{x:\rho(x,\mu)=0\},
\qquad D_x\rho\ne0.
\label{eq:stable-level-set}
\end{equation}
The sign of $\rho$ identifies the side of the boundary and therefore distinguishes the two local deterministic reactive fates.  The value of $\rho$ is not required to be a geometric distance; it is simply a convenient signed coordinate transverse to the boundary.

Let $\Phi_T^u(x_0,\mu)$ denote the endpoint at time $T$ of \eqref{eq:switch-control-system}.  Put $\delta=\mu-\mu_c$ and define the signed terminal boundary coordinate
\begin{equation}
F_T(\delta,u)
=
\rho\!\left(\Phi_T^u(x_0(\mu_c+\delta),\mu_c+\delta),\mu_c+\delta\right).
\label{eq:endpoint-function}
\end{equation}
The condition $F_T(0,0)=0$ says that at $\mu=\mu_c$ the unforced trajectory reaches the deterministic reactivity boundary at time $T$.  The underlying cost is the standard finite-time Freidlin--Wentzell control action \cite{FreidlinWentzell2012,ERenVandenEijnden2004}.  For a fixed sufficiently small radius $r>0$, we localize the minimization to the small-control branch and call the resulting quantity the \emph{local boundary-reaching action}:
\begin{equation}
S_T^{\rm loc}(\mu)
=
\inf\left\{
J_T[u]:
F_T(\mu-\mu_c,u)=0,
\ \|u\|_{L^2}\le r
\right\}.
\label{eq:local-boundary-action-def}
\end{equation}
The radius only restricts attention to the small-control branch that approaches zero as $\mu\to\mu_c$; the leading coefficient obtained below does not depend on this auxiliary choice.

Two first-order sensitivities determine the answer.  The derivative $a_T=\partial_\delta F_T(0,0)$ measures how rapidly the unforced terminal point moves across the boundary when the parameter is varied.  The Fr\'echet derivative with respect to the control,
\[
D_uF_T(0,0):L^2([0,T];\R^m)\longrightarrow\R,
\]
tells us how an infinitesimal forcing history changes the same terminal boundary coordinate.  By definition a Fr\'echet derivative is a bounded linear map.  The space $L^2([0,T];\R^m)$ is a Hilbert space, so the Riesz representation theorem implies that there is a unique function $k_T\in L^2([0,T];\R^m)$ such that
\begin{equation}
D_uF_T(0,0)v
=\int_0^T k_T(t)\cdot v(t)\,\dd t
\label{eq:Riesz-explained}
\end{equation}
for every control perturbation $v\in L^2$ \cite{Brezis2011}.  The function $k_T(t)$ is therefore a sensitivity kernel: it tells us how strongly a small forcing applied at time $t$ affects the terminal boundary coordinate.  Its squared norm
\[
Q_T=\int_0^T |k_T(t)|^2\,\dd t
\]
measures how effectively the available forcing can move the terminal point in the direction needed to reach the boundary.  The control-theoretic relation between $Q_T$ and the finite-time controllability Gramian is derived explicitly after Corollary~\ref{cor:adjoint-coefficient}.  The competition between $a_T$ and $Q_T$ produces the quadratic law stated in \Cref{thm:quadratic-switching}.

\begin{theorem}[Quadratic action near a deterministic reactivity boundary]
\label{thm:quadratic-switching}
Assume that $f$, $B$, $x_0$, and $\rho$ are $C^3$ near the critical trajectory.  Suppose
\begin{equation}
a_T:=\partial_\delta F_T(0,0)\ne0,
\label{eq:aT-def}
\end{equation}
so that changing the parameter moves the unforced endpoint across the boundary to first order.  Suppose also that the Fr\'echet derivative with respect to the control is nonzero.  Let $k_T\in L^2([0,T];\R^m)$ be the unique Riesz representative described in \eqref{eq:Riesz-explained}; that is,
\begin{equation}
D_uF_T(0,0)v
=\int_0^T k_T(t)\cdot v(t)\,\dd t,
\label{eq:Riesz-kernel}
\end{equation}
and define
\begin{equation}
Q_T=\int_0^T |k_T(t)|^2\,\dd t>0.
\label{eq:QT-def}
\end{equation}
Then the local minimum action required to reach $\mathcal W_\mu$ satisfies
\begin{equation}
S_T^{\rm loc}(\mu)
=\frac{a_T^2}{2Q_T}(\mu-\mu_c)^2
+O\!\left(|\mu-\mu_c|^3\right).
\label{eq:quadratic-law}
\end{equation}
Moreover, the minimizing control has the expansion
\begin{equation}
u_\mu(t)
=-\frac{a_T(\mu-\mu_c)}{Q_T}\,k_T(t)
+O\!\left((\mu-\mu_c)^2\right)
\label{eq:optimal-control-leading}
\end{equation}
in $L^2$.  If all other feasible boundary-reaching branches have action bounded away from zero as $\mu\to\mu_c$, then the local minimum in \eqref{eq:quadratic-law} is also the global minimum for $|\mu-\mu_c|$ sufficiently small.
\end{theorem}

The physical interpretation of the coefficient $a_T^2/(2Q_T)$ is clearest before turning to the proof.  If changing $\mu$ moves the deterministic endpoint rapidly away from the boundary, $|a_T|$ is large and more forcing is required to compensate.  If the available forcing moves the endpoint efficiently in the normal direction, $Q_T$ is large and the required action is smaller.  The quadratic power itself reflects the quadratic cost $\frac12\int|u|^2\,\dd t$ together with the fact that the displacement to be corrected is linear in $\mu-\mu_c$.

\begin{proof}
We first note that the minimum in \eqref{eq:local-boundary-action-def} is attained whenever the local feasible set is nonempty.  For sufficiently small $r$, the endpoint map is weakly sequentially continuous on bounded subsets of $L^2$ for the finite-dimensional control-affine system \eqref{eq:switch-control-system}; the feasible set in the closed ball is therefore weakly closed, while $J_T$ is weakly lower semicontinuous.  The direct method gives a minimizer.

Taylor expansion of \eqref{eq:endpoint-function} at $(\delta,u)=(0,0)$ gives
\begin{equation}
F_T(\delta,u)
=a_T\delta+\int_0^T k_T(t)\cdot u(t)\,\dd t
+R(\delta,u),
\label{eq:endpoint-expansion}
\end{equation}
where
\begin{equation}
|R(\delta,u)|
\le C\left(\delta^2+|\delta|\,\|u\|_{L^2}+\|u\|_{L^2}^2\right)
\end{equation}
near the origin.  Ignoring the remainder for the moment, the linearized terminal condition is
\begin{equation}
\int_0^T k_T(t)\cdot u(t)\,\dd t=-a_T\delta.
\label{eq:linear-terminal-condition}
\end{equation}
By the Cauchy--Schwarz inequality,
\begin{equation}
a_T^2\delta^2
\le Q_T\int_0^T |u(t)|^2\,\dd t,
\end{equation}
and equality is attained by
\begin{equation}
u_0(t)=-\frac{a_T\delta}{Q_T}k_T(t).
\end{equation}
The linearized minimum cost is therefore $a_T^2\delta^2/(2Q_T)$.

To pass to the nonlinear endpoint condition, introduce a scalar Lagrange multiplier $\nu\in\R$ and consider
\begin{equation}
u+\nu D_uF_T(\delta,u)^*=0,
\qquad
F_T(\delta,u)=0.
\label{eq:lagrange-system}
\end{equation}
At $(\delta,u,\nu)=(0,0,0)$ the derivative of \eqref{eq:lagrange-system} with respect to $(u,\nu)$ maps $(v,\beta)$ to
\begin{equation}
\left(v+\beta k_T,\;\langle k_T,v\rangle_{L^2}\right).
\end{equation}
This operator is invertible when $Q_T=\|k_T\|_{L^2}^2>0$: given $(f,g)$, one first finds $\beta=(\langle k_T,f\rangle-g)/Q_T$ and then $v=f-\beta k_T$.  The implicit-function theorem therefore gives a unique small stationary branch $(u_\delta,\nu_\delta)$ with
\begin{equation}
u_\delta
=-\frac{a_T\delta}{Q_T}k_T+O(\delta^2),
\qquad
\nu_\delta=O(\delta).
\end{equation}
Substitution into \eqref{eq:switch-cost} proves \eqref{eq:quadratic-law}.  The Hessian of the constrained Lagrangian in the control direction is the identity plus the $O(\delta)$ perturbation $\nu_\delta D_{uu}^2F_T$.  It is therefore positive definite on the tangent space of the constraint for sufficiently small $|\delta|$.  The stationary branch is consequently a strict local minimum near the origin.

It remains to connect this stationary branch with the minimum in \eqref{eq:local-boundary-action-def}.  The expansion above supplies a feasible control of norm $O(|\delta|)$ and cost $O(\delta^2)$.  Hence every minimizer in the closed ball satisfies $\|u\|_{L^2}=O(|\delta|)$ as $\delta\to0$.  For $|\delta|$ sufficiently small it therefore lies in the neighborhood in which $D_uF_T\ne0$ and the implicit-function theorem gives the unique stationary branch.  Every minimizer must satisfy the Lagrange-multiplier equations, so it coincides with $u_\delta$.  This proves the asserted local minimum formula.  The final statement follows because this branch has action tending to zero while any competing branch bounded away from zero cannot minimize sufficiently near $\mu_c$.
\end{proof}

The coefficient in \eqref{eq:quadratic-law} can be calculated without solving a separate optimization problem for several nearby parameter values.  The key observation is that one only needs to know how a small perturbation applied at each earlier time affects the signed boundary coordinate at the final time.  An adjoint equation transports the final normal covector backward along the deterministic trajectory and provides exactly this sensitivity.  Let
\begin{equation}
x_c(t)=\Phi_t^0(x_0(\mu_c),\mu_c)
\end{equation}
be the critical deterministic trajectory and define
\begin{equation}
A(t)=D_xf(x_c(t),\mu_c),
\qquad
B_c(t)=B(x_c(t),\mu_c).
\end{equation}
Let
\begin{equation}
n_T=D_x\rho(x_c(T),\mu_c)
\end{equation}
be a normal covector to the deterministic reactivity boundary at the terminal point.  The adjoint equation is
\begin{equation}
\dot\lambda=-A(t)^{\cT}\lambda,
\qquad
\lambda(T)=n_T.
\label{eq:adjoint-equation}
\end{equation}

\begin{corollary}[Computing the coefficient from the critical trajectory]
\label{cor:adjoint-coefficient}
With the notation above,
\begin{equation}
k_T(t)=B_c(t)^{\cT}\lambda(t),
\qquad
Q_T=\int_0^T|B_c(t)^{\cT}\lambda(t)|^2\,\dd t.
\label{eq:adjoint-Q}
\end{equation}
If $\eta(t)=\partial_\mu x(t;\mu,0)|_{\mu=\mu_c}$, then
\begin{equation}
\dot\eta=A(t)\eta+\partial_\mu f(x_c(t),\mu_c),
\qquad
\eta(0)=x_0'(\mu_c),
\label{eq:parameter-variational}
\end{equation}
and
\begin{equation}
a_T
=\partial_\mu\rho(x_c(T),\mu_c)+n_T\cdot\eta(T).
\label{eq:aT-adjoint}
\end{equation}
Consequently
\begin{equation}
C_T:=\frac{a_T^2}{2Q_T}
\label{eq:CT-coefficient}
\end{equation}
is determined entirely by the critical deterministic trajectory, the motion of the deterministic reactivity boundary with the parameter, and the effectiveness of the available noise in the direction normal to the reactivity boundary.
\end{corollary}

\begin{proof}
The linearized response to a control perturbation is
\begin{equation}
\dot\xi=A(t)\xi+B_c(t)u(t),
\qquad \xi(0)=0.
\end{equation}
To see the adjoint identity explicitly, differentiate the scalar pairing $\lambda(t)\cdot\xi(t)$:
\begin{align}
\frac{\dd}{\dd t}\bigl(\lambda\cdot\xi\bigr)
&=\dot\lambda\cdot\xi+\lambda\cdot\dot\xi\\
&=(-A^{\cT}\lambda)\cdot\xi
 +\lambda\cdot(A\xi+B_cu)\\
&=(B_c^{\cT}\lambda)\cdot u.
\label{eq:adjoint-pairing-derivative}
\end{align}
The two terms containing $A$ cancel because $(A^{\cT}\lambda)\cdot\xi=\lambda\cdot A\xi$.  Integrating \eqref{eq:adjoint-pairing-derivative} from $0$ to $T$, and using $\xi(0)=0$ together with $\lambda(T)=n_T$, gives
\begin{equation}
n_T\cdot\xi(T)
=\int_0^T B_c(t)^{\cT}\lambda(t)\cdot u(t)\,\dd t.
\label{eq:adjoint-pairing-integral}
\end{equation}
Comparison with the Riesz representation \eqref{eq:Riesz-kernel} shows that
\[
k_T(t)=B_c(t)^{\cT}\lambda(t),
\]
and therefore proves \eqref{eq:adjoint-Q}.  Differentiation of the uncontrolled trajectory with respect to $\mu$ gives \eqref{eq:parameter-variational}; differentiating the level-set condition gives \eqref{eq:aT-adjoint}.
\end{proof}

The control-theoretic meaning of $Q_T$ can now be stated precisely.  Let $\Phi_A(t,s)$ be the state-transition matrix for the linearized homogeneous equation $\dot\xi=A(t)\xi$.  The adjoint solution is
\[
\lambda(t)=\Phi_A(T,t)^{\cT}n_T,
\]
so
\[
k_T(t)=B_c(t)^{\cT}\Phi_A(T,t)^{\cT}n_T.
\]
The standard finite-time controllability Gramian for the linearized controlled system is \cite{Sontag1998}
\begin{equation}
W_c(T)
=
\int_0^T
\Phi_A(T,t)B_c(t)B_c(t)^{\cT}\Phi_A(T,t)^{\cT}\,\dd t.
\label{eq:controllability-gramian}
\end{equation}
Consequently
\begin{equation}
Q_T=n_T^{\cT}W_c(T)n_T.
\label{eq:projected-gramian}
\end{equation}
Thus $Q_T$ is not a new type of Gramian: it is the ordinary finite-time controllability Gramian evaluated in the one-dimensional terminal direction normal to the reactivity boundary.  This is the precise sense in which it is a projected controllability Gramian.

\begin{remark}[When the boundary itself moves]
The term $\partial_\mu\rho$ in \eqref{eq:aT-adjoint} accounts for motion of the deterministic reactivity boundary itself as the parameter changes.  This contribution is essential in the solvent application.  An equivalent computational device is to append $\mu$ as a frozen variable, $\dot\mu=0$.  The corresponding component of the adjoint covector then records the motion of the boundary automatically.
\end{remark}

\begin{remark}[Noise acting in only some variables]
No inverse diffusion matrix occurs in \Cref{thm:quadratic-switching}.  \Cref{thm:quadratic-switching} uses only the directions in which the noise acts, represented by $B$, and the response quantity $Q_T$.  It therefore remains valid when the diffusion matrix is singular, provided those available noise directions can influence the normal direction to the deterministic reactivity boundary at first order.  This is what is meant here by \emph{degenerate noise}.  In the solvent model below the random force acts directly on only one momentum variable.
\end{remark}

\section{Solvent inertia and a moving deterministic reactivity boundary}
\label{sec:solvent-application}

We now apply \Cref{thm:quadratic-switching} to a reaction-dynamics model introduced by Garcia-Meseguer and Carpenter and subsequently analyzed in phase space by Garcia-Meseguer, Carpenter, and Wiggins \cite{GarciaMeseguerCarpenter2019,GarciaMeseguerCarpenterWiggins2019}.  The original purpose of the model was to isolate an inertial solvent effect.  A reacting solute changes its shape on a time scale much shorter than the time required for a surrounding solvent shell to reorganize.  The important feature for the present problem is that the potential-energy surface (PES) is held fixed while the effective solvent mass is varied.  The published phase-space calculation showed that the periodic-orbit transition-state geometry nevertheless moves in phase space.  The model therefore separates a potential-energy barrier from a dynamical reactivity boundary in a particularly clean way.

As distinguished in \Cref{sec:switching-law}, the stochastic forcing introduced here is used only to measure the weak-noise action required to reach the deterministic phase-space reactivity boundary of the unforced Hamiltonian system; we neither construct a noise-realization-dependent stochastic transition state nor compute a committor.

The word ``Hamiltonian'' now refers to a different object from the one used in the first part of the paper.  The function $H_{\mu_2}$ introduced below is the \emph{physical mechanical Hamiltonian} of the solvent--solute system, whose state variables are $(r_1,r_2,p_1,p_2)$.  By contrast, $H_{\mathrm{FW}}$ in \Cref{sec:fw} is the auxiliary Hamiltonian obtained from a stochastic differential equation by introducing conjugate variables for a large-deviation variational problem.  If weak noise is added to this four-dimensional mechanical system, one could apply the Freidlin--Wentzell construction once again and obtain an eight-dimensional auxiliary Hamiltonian system.  We do not need that additional step here because the quadratic theorem was formulated directly with the control representation of the action.

The application gives three numerical checks of \Cref{thm:quadratic-switching}.  First we compute the minimum action for several solvent masses close to the deterministic threshold and test the predicted quadratic dependence.  Second we compute the coefficient independently using only the deterministic trajectory at the threshold mass.  Third we solve one nonlinear controlled problem to check the linear approximation used in deriving the local formula.

\subsection{The deterministic solvent--solute Hamiltonian}

The model has a reactive coordinate $r_1$, a solvent coordinate $r_2$, and conjugate momenta $p_1,p_2$.  The Hamiltonian is
\begin{equation}
H_{\mu_2}(r,p)
=\frac{p_1^2}{2\mu_1}+\frac{p_2^2}{2\mu_2}+V(r_1,r_2),
\label{eq:solvent-H}
\end{equation}
with $\mu_1=1$ and
\begin{equation}
V(r_1,r_2)
=\sum_{j=1}^{5}c_jr_1^{j-1}
+c_6(c_7-r_2)^2
+\frac{c_8}{(r_2-r_1)^{12}}.
\label{eq:solvent-V}
\end{equation}
The coefficients are taken directly from Table A.1 of Ref.~\cite{GarciaMeseguerCarpenterWiggins2019}:
\begin{center}
\begin{tabular}{c@{\qquad}c@{\qquad}c@{\qquad}c}
\toprule
$c_1$ & $c_2$ & $c_3$ & $c_4$\\
\midrule
$321.904484$ & $-995.713452$ & $1118.689573$ & $-537.856726$\\
\bottomrule
\end{tabular}
\end{center}
\begin{center}
\begin{tabular}{c@{\qquad}c@{\qquad}c@{\qquad}c}
\toprule
$c_5$ & $c_6$ & $c_7$ & $c_8$\\
\midrule
$92.976121$ & $1$ & $1$ & $0.01$\\
\bottomrule
\end{tabular}
\end{center}
Solving $\nabla V=0$ directly from these coefficients gives the reactant minimum
\begin{equation}
(r_1,r_2)=(0.98778676,1.80661210),
\qquad V_R=0.77040026,
\end{equation}
the product minimum
\begin{equation}
(r_1,r_2)=(1.98516984,2.75641794),
\qquad V_P=-6.66284435,
\end{equation}
and the potential-energy saddle
\begin{equation}
(r_1^\ddagger,r_2^\ddagger)=(1.36560508,2.16176887),
\qquad V^\ddagger=3.47291422.
\label{eq:solvent-saddle}
\end{equation}
We use the energy stated in Sec.~3 of Ref.~\cite{GarciaMeseguerCarpenterWiggins2019},
\begin{equation}
E=3.691966889,
\label{eq:solvent-energy}
\end{equation}
which is $0.21905267$ above the saddle energy.

For each value of $\mu_2$ considered below, the fixed-energy surface contains the unstable periodic orbit that forms the one-dimensional transition-state NHIM for this two-degree-of-freedom system.  For the present calculations, the published fits to these orbits were used only as initial guesses; the periodic orbits were then computed by shooting from the Hamiltonian equations.  The orbits for $\mu_2=10$ and $100$ are shown in \Cref{fig:solvent-ts-geometry}; the potential is identical in the two cases.

\begin{figure}[htbp]
\centering
\includegraphics[width=0.78\textwidth]{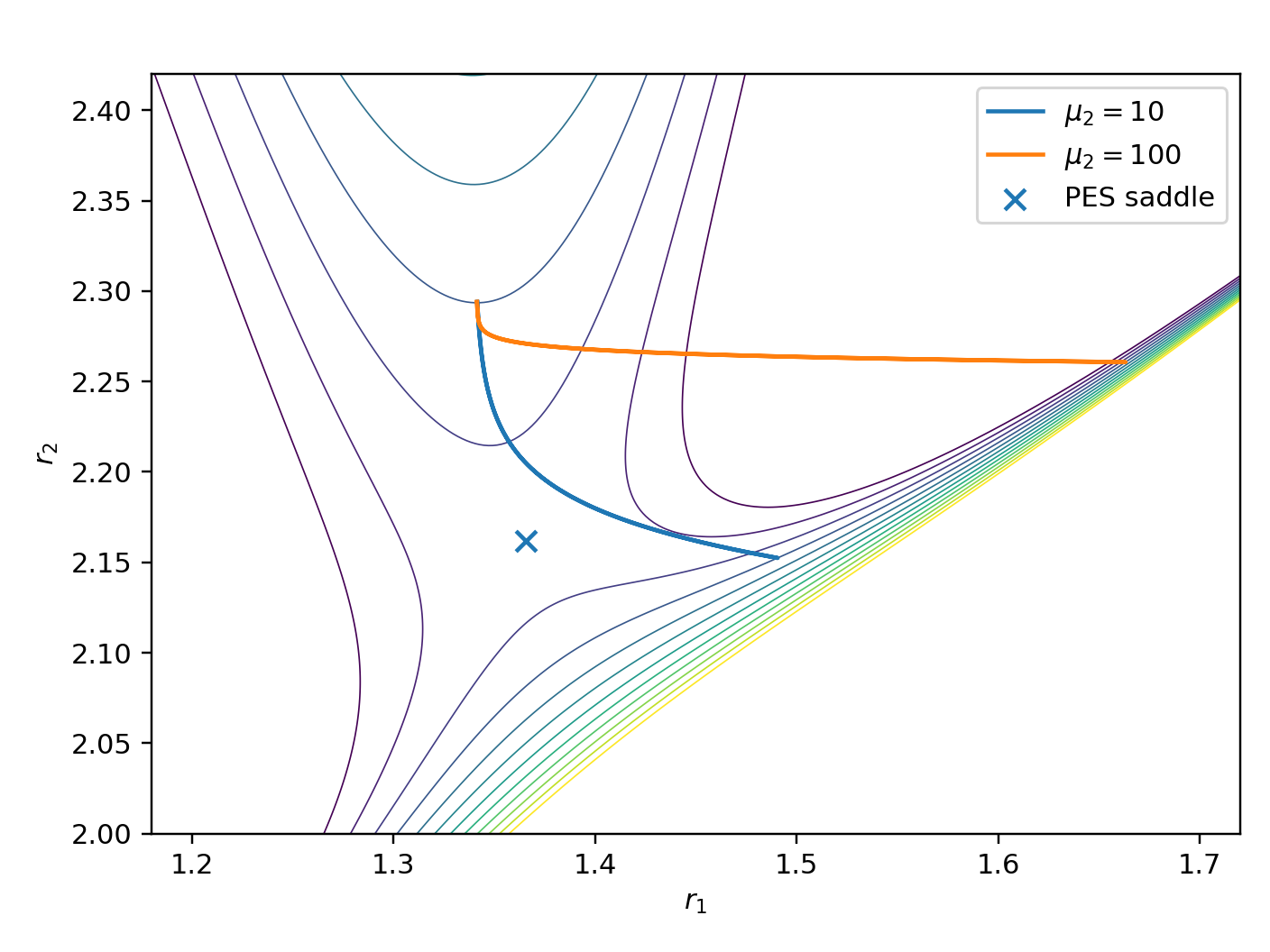}
\caption{Configuration-space projections of the transition-state periodic orbits at $\mu_2=10$ and $100$, computed here by shooting from zero-momentum turning points at the published total energy.  The potential-energy contours are identical because only the kinetic-energy mass parameter is changed.  The cross marks the PES saddle.}
\label{fig:solvent-ts-geometry}
\end{figure}

The linearization at the potential-energy saddle gives one real eigenvalue pair $\pm\lambda$ and one imaginary pair $\pm i\omega$.  The real pair describes exponential growth and decay in the hyperbolic direction; $\omega$ is the frequency of the bounded linear oscillation.  Computing these quantities from the Hessian of $V$ and the mass matrix gives
\begin{center}
\begin{tabular}{c@{\qquad}ccc}
\toprule
$\mu_2$ & $\lambda$ & $\omega$ & $\lambda/\omega$\\
\midrule
$0.1$ & $8.99869$ & $20.72097$ & $0.43428$\\
$1$   & $8.04736$ & $7.32716$  & $1.09829$\\
$10$  & $7.31539$ & $2.54889$  & $2.87003$\\
$100$ & $7.16174$ & $0.82332$  & $8.69857$\\
\bottomrule
\end{tabular}
\end{center}
As $\mu_2$ increases from $0.1$ to $100$, the hyperbolic rate decreases by about $20\%$, whereas the oscillation frequency decreases by roughly a factor of $25$.  The strong growth of the ratio $\lambda/\omega$ therefore comes mainly from the increasing oscillatory time scale, not from a comparable increase in the hyperbolic instability.

\subsection{A deterministic threshold between reactive fates}

We now choose a one-parameter family of initial conditions for which changing only the solvent mass changes the deterministic reactive fate.  We fix the initial configuration at the potential-energy saddle,
\[
(r_1(0),r_2(0))=(r_1^\ddagger,r_2^\ddagger),
\]
and choose the momenta subject to the fixed-energy constraint.  At this configuration the largest possible magnitude of $p_1$ occurs when $p_2=0$, so define
\begin{equation}
p_{1,\max}=\sqrt{2\mu_1(E-V^\ddagger)}.
\end{equation}
We then choose
\begin{equation}
p_1(0)=0.75\,p_{1,\max}.
\label{eq:p1-family}
\end{equation}
The solvent momentum is the negative solution of the energy constraint,
\begin{equation}
p_2(0;\mu_2)
=-\sqrt{2\mu_2\left(E-V^\ddagger-\frac{p_1(0)^2}{2\mu_1}\right)}.
\label{eq:p2-family}
\end{equation}
Only the solvent mass is varied.  For definiteness, a trajectory is recorded as reaching the product side when $r_1$ first reaches the product-minimum value $1.98516984$, and as returning to the reactant side when $r_1$ first reaches the reactant-minimum value $0.98778676$.

Bisection of the deterministic Hamiltonian trajectories gives
\begin{equation}
\boxed{\mu_c\approx9.975918.}
\label{eq:mu-critical}
\end{equation}
This is not a universal critical solvent mass.  It is the threshold value for the particular initial-condition family \eqref{eq:p1-family}--\eqref{eq:p2-family} and the stated first-arrival convention.  At this value the selected deterministic trajectory lies on the local stable manifold separating the two deterministic reactive fates.  For masses on opposite sides of $\mu_c$ the same construction produces different deterministic reactive fates, even though the potential-energy surface is unchanged.

\subsection{Weak forcing of the solvent momentum}

We add a weak force only to the solvent momentum.  In control form,
\begin{align}
\dot r_1&=p_1,\\
\dot r_2&=\frac{p_2}{\mu_2},\\
\dot p_1&=-\partial_{r_1}V,\\
\dot p_2&=-\partial_{r_2}V+u(t),
\label{eq:solvent-control}
\end{align}
with action
\begin{equation}
I_T[u]=\frac12\int_0^T u(t)^2\,\dd t.
\label{eq:solvent-action}
\end{equation}
Equivalently, the stochastic equation has additive noise $\sqrt{\eps}\,\dd W_t$ in the $p_2$ equation.  A constant noise amplitude $\sigma$ rescales the action coefficient by $1/\sigma^2$.

The random force acts directly only in the $p_2$ equation.  Thus the diffusion is singular in the four-dimensional phase space: this is the concrete instance of degenerate noise mentioned above.

There is one further geometric point.  At a single mechanical energy, the transition-state periodic orbit is one-dimensional and its stable manifold lies inside the three-dimensional energy surface.  The forcing in \eqref{eq:solvent-control} changes mechanical energy, so a forced trajectory is not confined to that surface.  We therefore require a deterministic reactivity boundary in the full four-dimensional mechanical phase space.

Take a smooth family of transition-state periodic orbits over a small interval of energies.  If $N_{\mu_2,E'}$ denotes the periodic orbit at energy $E'$, define
\begin{equation}
\mathcal N_{\mu_2,I}
=
\bigcup_{E'\in I}N_{\mu_2,E'}.
\label{eq:energy-cylinder}
\end{equation}
One dimension of $\mathcal N_{\mu_2,I}$ is the phase along the periodic orbit and the second is the energy parameter, so $\dim\mathcal N_{\mu_2,I}=2$.  Provided the family is smooth and uniformly normally hyperbolic for $E'\in I$, this union is a locally invariant normally hyperbolic cylinder.  It has one stable direction, hence
\[
\dim W^s(\mathcal N_{\mu_2,I})=3.
\]
Its stable manifold is therefore codimension one in the four-dimensional mechanical phase space and locally separates the two deterministic reactive fates.  This is the \emph{deterministic reactivity boundary} used in the minimum-action calculation.  It is an invariant structure of the unforced Hamiltonian dynamics; it is not the noise-dependent moving surface of stochastic transition-state theory and it is not defined by a committor.

\subsection{Why the equilibrium energy barrier does not detect the effect}

There is a useful comparison with an equilibrium Langevin formulation.  If friction and fluctuation satisfying fluctuation--dissipation balance are added in the solvent momentum, then the mechanical-energy difference $W=H-H_R$ satisfies the stationary zero-energy Hamilton--Jacobi equation.  When the additional reachability hypotheses needed to identify this solution with the quasipotential hold, the leading saddle barrier is
\begin{equation}
V^\ddagger-V_R=2.70251397,
\end{equation}
which is independent of $\mu_2$.  Thus the mass-dependent displacement of the dynamical transition state is not encoded in this equilibrium energy difference.  This is why the finite-time boundary-reaching problem, rather than an equilibrium barrier calculation, is used below.  The distinction is consistent with the different roles played by Kramers/Grote--Hynes rate corrections and by a no-recrossing phase-space dividing surface \cite{Kramers1940,GroteHynes1980,GarciaMeseguerCarpenter2019}.

\subsection{First check: quadratic scaling away from the threshold}

For each $\mu_2>\mu_c$ in \Cref{tab:stable-action}, the transition-state periodic orbit was computed by shooting and the unforced initial condition was integrated until its first return to the reactant-side threshold.  Near the point where this trajectory comes closest to the periodic orbit, the stable manifold can be approximated to first order by its tangent hyperplane.  A left unstable Floquet covector supplies the normal covector to this hyperplane.  Integrating the adjoint equation backward then tells us how a small solvent force applied at each earlier time changes the final normal displacement.  From this response we obtain the minimum linearized action in \Cref{thm:quadratic-switching}.

At each off-threshold mass, let $d_T(\mu_2)$ be the signed normal displacement of the uncontrolled endpoint from this tangent hyperplane at the chosen closest-approach section.  The corresponding linearized boundary-reaching action is
\begin{equation}
S_{W^s}^{\rm lin}(\mu_2)
=
\frac{d_T(\mu_2)^2}{2Q_T(\mu_2)}.
\label{eq:SWs-lin-def}
\end{equation}
Near the threshold,
\[
d_T(\mu_2)=a_T(\mu_2-\mu_c)+O((\mu_2-\mu_c)^2),
\]
so \Cref{thm:quadratic-switching} predicts the quadratic limiting coefficient $a_T^2/(2Q_T)$.  The values in \Cref{tab:stable-action} are generated directly from the model coefficients by the standalone computation described in Appendix~A.
\begin{table}[htbp]
\centering
\caption{Computed linearized minimum action to the local deterministic reactivity boundary.}
\label{tab:stable-action}
\begin{tabular}{cc@{\qquad}cc}
\toprule
$\mu_2$ & $S_{W^s}^{\rm lin}$ & $\mu_2$ & $S_{W^s}^{\rm lin}$\\
\midrule
9.9762 & $7.62\times10^{-10}$ & 9.9900 & $1.898\times10^{-6}$\\
9.9770 & $1.121\times10^{-8}$ & 10.0000 & $5.550\times10^{-6}$\\
9.9780 & $4.148\times10^{-8}$ & 10.0200 & $1.860\times10^{-5}$\\
9.9800 & $1.594\times10^{-7}$ & 10.0500 & $5.251\times10^{-5}$\\
\bottomrule
\end{tabular}
\end{table}
With the value of $\mu_c$ obtained by the bisection above, a two-parameter least-squares fit in logarithmic coordinates gives
\begin{equation}
S_{W^s}^{\rm lin}
\approx0.00957591\,(\mu_2-\mu_c)^{2.00019}.
\label{eq:solvent-free-fit}
\end{equation}
The exponent is numerically indistinguishable from the value two predicted by \Cref{thm:quadratic-switching}.  Since \Cref{thm:quadratic-switching} fixes the exponent, the more appropriate coefficient comparison holds the exponent at two.  We use the log-coordinate estimate
\begin{equation}
C_2
=\exp\left[
\frac1N\sum_{j=1}^{N}
\log\left(\frac{S_j}{(\mu_{2,j}-\mu_c)^2}\right)
\right]
=0.00956668.
\label{eq:solvent-fixed-fit}
\end{equation}
\Cref{fig:solvent-action-scaling} shows the computed actions together with the quadratic laws obtained from the fitted coefficient $C_2$ and the independently determined critical-orbit coefficient $C_{\rm direct}$ discussed below.
The smallest action in \Cref{tab:stable-action} is of order $10^{-10}$ and is the most numerically delicate point; it is included in the displayed fit, but no separate scaling claim is inferred from that point alone.

\begin{figure}[htbp]
\centering
\includegraphics[width=0.70\textwidth]{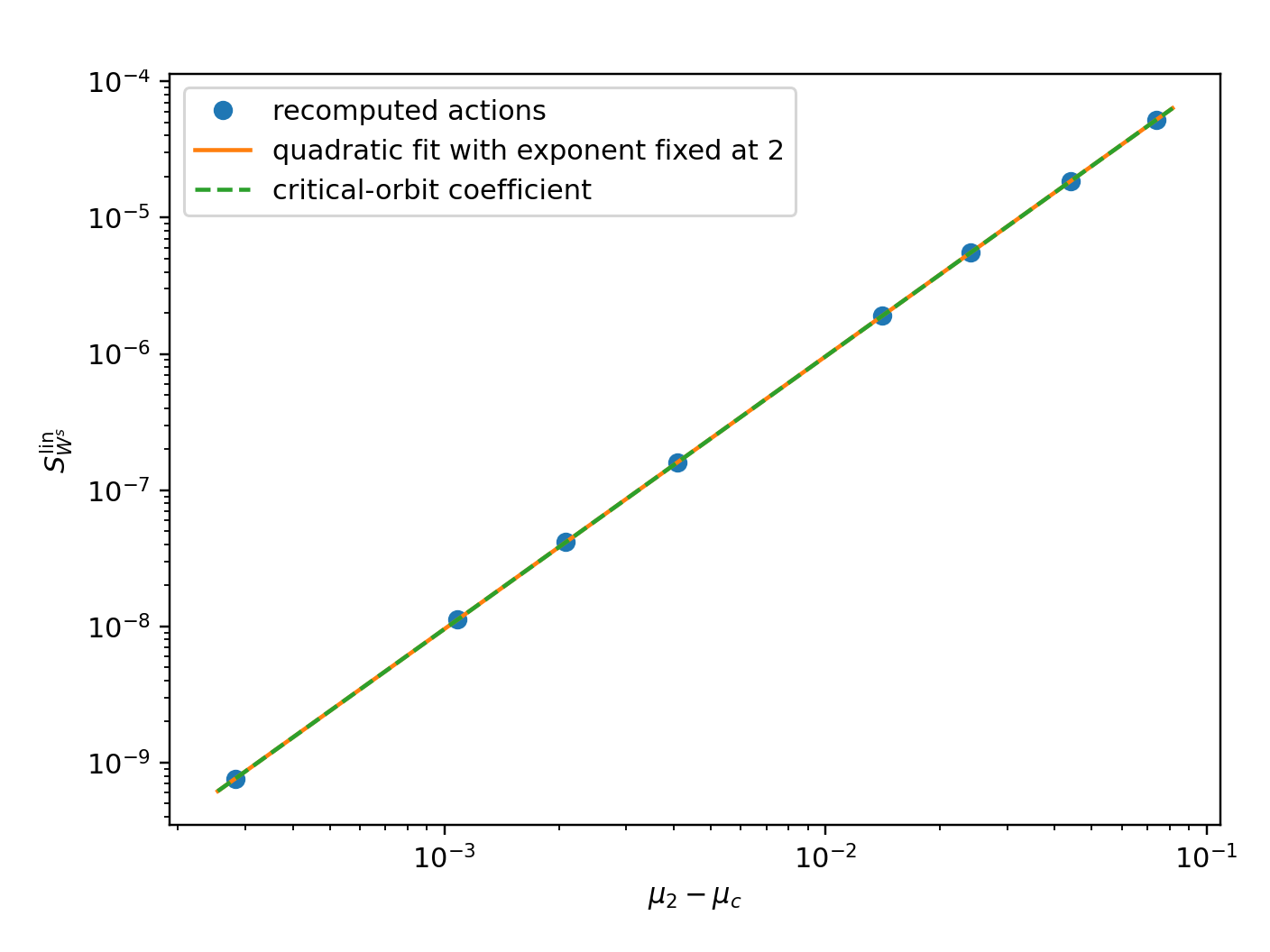}
\caption{Computed minimum linearized action for reaching the local deterministic reactivity boundary.  The horizontal coordinate is the distance from the computed deterministic threshold mass.  The solid quadratic reference uses the fixed-exponent coefficient $C_2$ in \eqref{eq:solvent-fixed-fit}; the dashed quadratic reference uses the critical-orbit coefficient $C_{\rm direct}$ in \eqref{eq:C-direct}.  The figure is generated directly from the numerical CSV supplied with the supplementary calculation package.}
\label{fig:solvent-action-scaling}
\end{figure}

The mass-weighted closest distance from the uncontrolled trajectory to the periodic orbit decreases approximately as $(\mu_2-\mu_c)^{0.52}$ over this short scan.  We do not use this exponent in \Cref{thm:quadratic-switching}.  The separate numerical values assigned to a Floquet coordinate and to its Gramian depend on the normalization and on the terminal section; the invariant object is their combination $a_T^2/(2Q_T)$, or equivalently the minimum action itself.  Accordingly, we do not assign separate scaling laws to these convention-dependent quantities.

\subsection{Second check: the coefficient from the threshold trajectory}

The first check uses several masses near $\mu_c$.  \Cref{thm:quadratic-switching} makes a stronger prediction: its leading coefficient can be obtained from the single deterministic trajectory at $\mu_c$, without fitting any off-threshold action values.  At the critical mass obtained above, shooting gives
\begin{equation}
T_{\rm per}=2.42823896,
\end{equation}
and integration of the variational equations gives
\begin{equation}
\Lambda_u\approx4.99427\times10^7,
\qquad
\lambda_u=\frac{\log\Lambda_u}{T_{\rm per}}\approx7.30010.
\label{eq:critical-floquet}
\end{equation}
The coefficient contains two pieces: the rate at which changing $\mu_2$ moves the trajectory relative to the stable-manifold boundary, and the effectiveness of solvent forcing in moving the trajectory across that boundary.  To include the motion of the transition-state geometry itself, we append $\mu_2$ as a frozen state, $\dot\mu_2=0$.  The left unstable Floquet covector of this extended system is then propagated backward along the threshold trajectory.  This single adjoint calculation gives both quantities entering $a_T^2/(2Q_T)$.  Four terminal sections give the values in \Cref{tab:critical-coefficient}.  The sign of $a_T$ depends only on the orientation chosen for the normal covector, so the table lists $|a_T|$.
\begin{table}[htbp]
\centering
\caption{Critical-orbit calculation of the quadratic coefficient, obtained without fitting the off-threshold action values.}
\label{tab:critical-coefficient}
\begin{tabular}{cccc}
\toprule
$T$ & $|a_T|$ & $Q_T$ & $a_T^2/(2Q_T)$\\
\midrule
1.50 & 100.30239 & 526056.54 & 0.00956225170447\\
1.55 & 100.30236 & 526056.18 & 0.00956225170436\\
1.60 & 100.30116 & 526043.62 & 0.00956225170429\\
1.65 & 100.30244 & 526057.01 & 0.00956225170433\\
\bottomrule
\end{tabular}
\end{table}
Thus
\begin{equation}
\boxed{C_{\rm direct}=0.00956225.}
\label{eq:C-direct}
\end{equation}
This calculation uses the critical deterministic orbit and adjoint equations only.  Comparing it with the theorem-constrained coefficient $C_2=0.00956668$ from \eqref{eq:solvent-fixed-fit} gives a relative difference of approximately $0.046\%$.  The small systematic excess of the off-critical ratios $S_j/(\mu_{2,j}-\mu_c)^2$ over $C_{\rm direct}$ is consistent with the use of a finite closest-approach section, which truncates the projected Gramian relative to the section-converged critical calculation.  For completeness, the unconstrained two-parameter fit \eqref{eq:solvent-free-fit} gives $C=0.00957591$ and exponent $2.00019$; the corresponding prefactor differs from $C_{\rm direct}$ by about $0.14\%$ because the fitted exponent and prefactor trade against one another.  The fixed-exponent comparison is therefore the primary coefficient check.

\subsection{Third check: a nonlinear controlled-flow calculation}

The calculations in \Cref{tab:stable-action} and \eqref{eq:C-direct} use the linearized endpoint response that enters \Cref{thm:quadratic-switching}.  To test the remaining linearization in the controlled dynamics, we performed a separate direct-collocation calculation at
\begin{equation}
\mu_2=\mu_c+0.05=10.02591761.
\end{equation}
The terminal time and periodic-orbit phase are fixed at the closest-approach section of the corresponding uncontrolled trajectory.  At that section the local stable-manifold boundary is represented by its unstable-Floquet tangent hyperplane.  The state equations \eqref{eq:solvent-control} themselves are not linearized: they are imposed with the full nonlinear potential \eqref{eq:solvent-V}, while the piecewise-constant solvent control is optimized to minimize \eqref{eq:solvent-action} subject to the terminal hyperplane constraint.

The linearized adjoint formula at the same section gives
\begin{equation}
S_{W^s}^{\rm lin}=2.3923834\times10^{-5}.
\end{equation}
Direct collocation gives
\begin{center}
\begin{tabular}{ccc}
\toprule
control intervals & nonlinear action & terminal residual from independent reintegration\\
\midrule
40  & $2.3094007\times10^{-5}$ & $1.75\times10^{-9}$\\
80  & $2.3882554\times10^{-5}$ & $9.12\times10^{-11}$\\
160 & $2.3921588\times10^{-5}$ & $5.06\times10^{-12}$\\
\bottomrule
\end{tabular}
\end{center}
As the number of control intervals is increased from $40$ to $80$ to $160$, the discretized nonlinear action approaches the linearized value.  At the finest discretization used here, $160$ intervals, the relative difference is $9.39\times10^{-5}$, or about $0.0094\%$.  We report this as the residual discrepancy at that discretization and do not claim that the exact continuum optimum has been reached.  This calculation is deliberately described as a controlled-flow check rather than as a global computation of the curved stable manifold: the terminal boundary is still represented by its local tangent hyperplane.  It therefore tests the linearization in the control-to-endpoint dynamics, which is the additional numerical check needed for \Cref{thm:quadratic-switching}, without introducing a separate global stable-manifold approximation.

Let $\mathcal E_T$ denote the local event that the noisy path reaches the deterministic reactivity boundary by time $T$.  This is the event computed by the minimum-action problem above.  At finite noise, crossing a deterministic invariant boundary is not by itself identical to probabilistic commitment to product, because subsequent noise can alter the path again; the committor is the natural object for that probabilistic question \cite{EVandenEijnden2006,EVandenEijnden2010}.  When the Freidlin--Wentzell upper and lower bounds match for the boundary-reaching event,
\begin{equation}
\lim_{\eps\to0}\eps\log\Pr(\mathcal E_T)=-S_T(\mu_2).
\end{equation}
Consequently \Cref{thm:quadratic-switching} predicts, locally,
\begin{equation}
\eps\log\Pr(\mathcal E_T)
\sim -C(\mu_2-\mu_c)^2
\end{equation}
on the logarithmic small-noise scale.  A uniform analysis of the joint limit $|\mu_2-\mu_c|=O(\sqrt\eps)$ would require a separate moderate-deviation argument and is not claimed.

\section{Discussion}
\label{sec:discussion}

The paper began with two questions that involve NHIMs in different ways.

The first question asked how the Freidlin--Wentzell fluctuation geometry is organized near a deterministic NHIM $N$.  The Freidlin--Wentzell construction enlarges the state from $x$ to $(x,p)$.  Near the zero-momentum copy $N_0$, the $k$ tangent directions of $N$ are accompanied by $k$ conjugate momentum directions.  This produces $2k$ center directions.  Locally, these directions are tangent to a symplectic invariant manifold having the same local symplectic structure as $T^*N$.  This center manifold describes the nearby Hamiltonian geometry, but it is not unique away from $N_0$.

The Hamiltonian fluctuation extremals that approach $N_0$ in backward time at the strong normal rate are organized by a different invariant object.  The union $W^{uu}(N_0)$ of the strong-unstable fibers has dimension $n$ and carries a single-valued action function.  In symplectic terminology, it is exact Lagrangian.  This statement matters because the projection of $W^{uu}(N_0)$ back to the original state space organizes the multiplicity of optimal fluctuation paths.  Where the projection folds, several extremals can reach the same state.  Where two minimizing action branches have equal value, the preferred fluctuation mechanism switches.

The counterexample clarifies a limitation that is easy to miss if one looks only at the deterministic system.  Normal hyperbolicity of $N$ is measured along trajectories on $p=0$.  Hamiltonian trajectories on a zero-energy section can have $p\ne0$, and this auxiliary momentum changes the motion in $x$.  Consequently the transverse rates sampled by those trajectories can differ from the deterministic rates.  Normal hyperbolicity of $N$ therefore organizes the local geometry near $N_0$ but does not, by itself, guarantee normal hyperbolicity of an entire zero-energy section within $H_{\mathrm{FW}}^{-1}(0)$.

The second question concerned a deterministic reactivity boundary supplied, in the reaction application, by the stable manifold of transition-state geometry.  If parameter variation changes the relative position of the deterministic trajectory and this boundary so that they meet at $\mu=\mu_c$, then a small fluctuation can reach the boundary when $\mu$ is close to $\mu_c$.  The minimum action required to reach the boundary satisfies
\[
S_T^{\rm loc}(\mu)
=\frac{a_T^2}{2Q_T}(\mu-\mu_c)^2+O(|\mu-\mu_c|^3).
\]
The coefficient has a direct interpretation.  The quantity $a_T$ measures how rapidly parameter variation moves the deterministic endpoint across the boundary.  The quantity $Q_T$ measures how effectively the available forcing can move the endpoint in the required normal direction.  Thus the boundary-reaching cost increases when the deterministic trajectory moves away from the boundary rapidly and decreases when the noise acts efficiently in the direction needed to cross it.  An adjoint calculation along the threshold trajectory provides these sensitivities without requiring a family of off-threshold optimization problems.

The solvent-inertia model makes the distinction between a potential-energy barrier and a deterministic phase-space reactivity boundary concrete.  The potential-energy surface is unchanged as the solvent mass varies, but the transition-state periodic orbit and its stable manifold move in phase space.  The equilibrium energy difference is therefore independent of solvent mass, whereas the finite-time action required to change the deterministic reactive fate is not.  The numerical calculations verify both the predicted quadratic dependence and the coefficient computed from the threshold trajectory, and the nonlinear controlled calculation checks the local linear approximation used in deriving \Cref{thm:quadratic-switching}.

The control formulation is also important because the random force need not act in every state variable.  In the solvent model it acts directly only on the solvent momentum, so the diffusion matrix is singular.  The quadratic law nevertheless applies because what matters is whether the available forcing can influence the normal direction to the deterministic reactivity boundary.  This suggests a natural extension to other systems in which environmental forcing acts through only a few collective coordinates.

\section{Conclusions}

Freidlin--Wentzell theory enlarges an $n$-dimensional stochastic state space to a $2n$-dimensional Hamiltonian phase space.  When the deterministic dynamics contains a $k$-dimensional NHIM $N$, this enlargement does more than copy $N$ into the zero-momentum set.  Near $N_0=N\times\{0\}$ there are $2k$ center directions, and these directions generate a local symplectic invariant manifold with cotangent-bundle geometry.  The Hamiltonian fluctuation extremals that approach $N_0$ in backward time at the strong normal rate lie in the strong-unstable manifold $W^{uu}(N_0)$, which carries their action as a single-valued function.  Normal hyperbolicity of the original $N$, however, need not make an entire zero-energy section normally hyperbolic within $H_{\mathrm{FW}}^{-1}(0)$ because nonzero conjugate momentum changes the trajectories along which transverse rates are sampled.

For reaction dynamics, a different question arises when a codimension-one deterministic reactivity boundary separates two reactive fates.  Near a parameter value at which the deterministic trajectory meets this boundary, the minimum weak-noise action needed to reach the boundary is quadratic in the distance from the threshold parameter.  The coefficient compares two effects: how quickly parameter variation moves the trajectory away from the boundary and how effectively the available noise can push it back toward the boundary.  Both can be obtained from the deterministic trajectory at threshold through an adjoint calculation, even when the noise acts directly in only some variables.

The solvent-inertia example shows why the distinction between a potential-energy barrier and a phase-space reactivity boundary matters physically.  Changing the solvent mass leaves the potential-energy surface fixed but moves the deterministic phase-space reactivity boundary.  The equilibrium energy barrier therefore remains unchanged, while the finite-time large-deviation cost of reaching the deterministic reactivity boundary varies.  In this sense the rare-event problem is organized by the dynamics of the phase-space boundary, not by the potential-energy saddle alone.

\appendix

\section{Reproducibility of the solvent-inertia calculation}
\label{app:solvent-numerics}

The numerical values in \Cref{sec:solvent-application} were computed directly from the Hamiltonian \eqref{eq:solvent-H}.  The supplementary calculation package contains standalone Python scripts that generate the reported values and CSV files.  The calculation is divided into three steps so that each can be checked separately.

First, the stationary points, the hyperbolic rate and oscillation frequency at the saddle, and the deterministic reactivity threshold are obtained directly from the published potential.  The threshold is found by bisection using the first-arrival convention stated in the main text.  Second, the transition-state periodic orbits are reconstructed by shooting from zero-momentum turning points.  The monodromy matrix and left unstable Floquet covector are obtained by integrating the variational equations along the periodic orbit.  Third, for each off-critical mass the unforced initial condition is integrated to its first return, its closest approach to the transition-state orbit is located, and the adjoint equation is integrated backward to obtain the boundary-normal Gramian quantity and the linearized boundary-reaching action.

The file \texttt{recompute\_solvent\_results.py} carries out the full calculation from the model coefficients.  Because that full scan is comparatively slow, the package also contains shorter scripts for the periodic-orbit reconstruction and the direct critical-orbit coefficient.  The CSV files used for \Cref{tab:stable-action,tab:critical-coefficient} are included.  No precomputed binary trajectory files are required.

The package also contains \texttt{nonlinear\_minimum\_action.py}, which independently solves the nonlinear controlled-flow check reported in the main text using direct collocation with 40, 80, and 160 control intervals and then reintegrates the optimized controls with a high-accuracy adaptive solver to verify the terminal constraint.  Separately normalized values of the unstable Floquet coordinate and its Gramian depend on normalization and terminal-section conventions; the action and the coefficient $a_T^2/(2Q_T)$ are invariant under rescaling of the normal covector, and these are the quantities reported in the main text.

The published input data are those of Ref.~\cite{GarciaMeseguerCarpenterWiggins2019}: Table A.1 gives the eight coefficients used in \eqref{eq:solvent-V}, and Sec.~3 gives $E=3.691966889$.

\section{Dimension summary}

It is useful to collect the dimensions of the principal geometric objects in one place.  If the physical state space has dimension $n$ and the deterministic NHIM has dimension $k$, the construction in \Cref{sec:carrier,sec:lagrangian} gives
\begin{center}
\begin{tabular}{lc}
\toprule
Object & Dimension\\
\midrule
$N_0=N\times\{0\}$ & $k$\\
Hamiltonian center bundle $E_H^c$ & $2k$\\
Local symplectic center manifold $\C_{\mathrm{FW}}$ & $2k$\\
Regular zero-energy section $\Sigma_{\mathrm{FW}}$ & $2k-1$\\
Strong-unstable manifold of fluctuation extremals $\Lagr_N^+$ & $n$\\
\bottomrule
\end{tabular}
\end{center}

In the two-degree-of-freedom reaction application the physical phase space is four-dimensional.  At fixed energy the transition-state NHIM is a one-dimensional unstable periodic orbit.  Allowing a small energy interval gives the two-dimensional cylinder $\mathcal N_{\mu_2,I}$ in \eqref{eq:energy-cylinder}, and its stable manifold is three-dimensional.  It is this codimension-one stable manifold that serves as the deterministic reactivity boundary for the degenerate-noise boundary-reaching problem.

\end{document}